\documentclass[sigconf, nonacm]{acmart}
\usepackage[linesnumbered,ruled,vlined]{algorithm2e}
\usepackage{epsfig,graphics,enumitem,float,subfigure}
\usepackage{epstopdf}
\usepackage{balance}
\usepackage{color}
\usepackage{setspace, url}
\usepackage{booktabs,multirow,multicolrule,makecell}

\newcommand\vldbdoi{XX.XX/XXX.XX}

\newcommand\vldbavailabilityurl{https://github.com/qq-dai/kplexEnum}
\newcommand\vldbpagestyle{plain}

\newcommand{\ignore}[1]{}
\newcommand{\nop}[1]{}
\newcommand{\eat}[1]{}
\newcommand{\kw}[1]{{\ensuremath{\mathsf{#1}}}\xspace}

\newcommand{\stitle}[1]{\vspace{1ex} \noindent{\bf #1}}
\long\def\comment#1{}

\newtheorem{definition}{Definition}

\newtheorem{example}{Example}

\newtheorem{theorem}{Theorem}
\newtheorem{lemma}{Lemma}

\newcommand{\titleEmph}[1]{\vspace{0.5ex} \noindent{\underline{\textit{#1}}}}

\newcommand{\fastPlex}{\emph{FP}\xspace}

\newcommand{\commuPlex}{\emph{CPlex}\xspace}
\newcommand{\dtwok}{\emph{D2K}\xspace}
\newcommand{\bran}{\kw{Branch}}

\newcommand{\FPBsc}{\emph{FP}{$\backslash$}\emph{ubs}\xspace}

\newcommand{\socEps}{Epinions}

\newcommand{\mrdata}[1]{\multirow{#1}{*}{\makecell[c]{Dataset\\($n$,$m$)}}}
\newcommand{\mrk}[1]{\multirow{#1}{*}{$k$}}
\newcommand{\mrq}[1]{\multirow{#1}{*}{$q$}}
\newcommand{\mrplex}[1]{\multirow{#1}{*}{\#$k$-plexes}}
\newcommand{\mcltime}[1]{\multicolumn{#1}{c|}{Running time (sec)}}
\newcommand{\mctime}[1]{\multicolumn{#1}{c}{Running time (sec)}}

\newcommand{\mrow}[2]{\multirow{#1}{*}{#2}}
\newcommand{\mrEml}[1]{\multirow{#1}{*}{\makecell[c]{EmEuAll\\(265214,\\420045)}}}
\newcommand{\mrSdot}[1]{\multirow{#1}{*}{\makecell[c]{Slashdot\\(82144,\\500480)}}}
\newcommand{\mrWiki}[1]{\multirow{#1}{*}{\makecell[c]{WikiVote\\(8298,\\100761)}}}
\newcommand{\mrEps}[1]{\multirow{#1}{*}{\makecell[c]{Epinions\\(75879,\\508837)}}}
\newcommand{\mrPoc}[1]{\multirow{#1}{*}{\makecell[c]{Pokec\\(1632803,\\30622564)}}}

\newcommand{\mrDBLP}[1]{\multirow{#1}{*}{\makecell[c]{DBLP\\(317080,\\1049866)}}}
\newcommand{\mrCai}[1]{\multirow{#1}{*}{\makecell[c]{Caida\\(26475,\\53381)}}}

\begin{document}
	\title{Scaling Up Maximal $k$-plex Enumeration
	}
	
	\author{Qiangqiang Dai}
	\affiliation{%
		\institution{Beijing Institute of Technology}
		\city{Beijing}
		\country{China}
	}
	\email{qiangd66@gmail.com}
	
	\author{Rong-Hua Li}
	\affiliation{%
		\institution{Beijing Institute of Technology}
		\city{Beijing}
		\country{China}
	}
	\email{lironghuabit@126.com}
	
	\author{Hongchao Qin}
	\affiliation{%
		\institution{Beijing Institute of Technology}
		\city{Beijing}
		\country{China}
	}
	\email{qhc.neu@gmail.com}
	
	\author{Meihao Liao}
	\affiliation{%
		\institution{Beijing Institute of Technology}
		\city{Beijing}
		\country{China}
	}
	\email{mhliao@bit.edu.cn}
	
	\author{Guoren Wang}
	\affiliation{%
		\institution{Beijing Institute of Technology}
		\city{Beijing}
		\country{China}
	}
	\email{wanggrbit@126.com}

	\begin{abstract}
		Finding all maximal $k$-plexes on networks is a fundamental research problem in graph analysis due to many important applications, such as community detection, biological graph analysis, and so on. A $k$-plex is a subgraph in which every vertex is adjacent to all but at most $k$ vertices within the subgraph. In this paper, we study the problem of enumerating all large maximal $k$-plexes of a graph and develop several new and efficient techniques to solve the problem. Specifically, we first propose several novel upper-bounding techniques to prune unnecessary computations during the enumeration procedure. We show that the proposed upper bounds can be computed in linear time. Then, we develop a new branch-and-bound algorithm with a carefully-designed pivot re-selection strategy to enumerate all $k$-plexes, which outputs all $k$-plexes in $O(n^2\gamma_k^n)$ time theoretically, where $n$ is the number of vertices of the graph and $\gamma_k$ is strictly smaller than 2. \color{black}{In addition, a parallel version of the proposed algorithm is further developed to scale up to process large real-world graphs. Finally, extensive experimental results show that the proposed sequential algorithm can achieve up to $2\times$ to $100\times$ speedup over the state-of-the-art sequential algorithms on most benchmark graphs. The results also demonstrate the high scalability of the proposed parallel algorithm. For example, on a large real-world graph with more than 200 million edges, our parallel algorithm can finish the computation within two minutes, while the state-of-the-art parallel algorithm cannot terminate within 24 hours.}
	\end{abstract}
	
	\maketitle
	
	\pagestyle{\vldbpagestyle}
	
	\section{Introduction} \label{sec:introduction}
	Graphs are ubiquitous in real-world applications. It is often of great value to find cohesive subgraphs from a graph in both theoretical research and practical applications. A classic model of cohesive subgraph called clique, where vertices are pairwise connected, has attracted much attention in recent decades. Many algorithms of enumerating all maximal cliques in a graph were developed, such as the classic Bron-Kerbosch algorithm \cite{cacmBronK73} and its pivot-based variants \cite{TomitaTT06worstcaseclique,13EppsteinLS,tcsNaude16}, the output-sensitive algorithms \cite{algorithmicaChangYQ13,icalpConteGMV16}, the parallel algorithms \cite{edbtConte16,SegundoPrallelClique18,tpdsWeiCT21}, and so on.
	
	Real-world networks are often noisy or faulty \cite{ConteMSGMV18,ZhouXGXJ20}. It may be overly restrictive to find cohesive subgraphs when using the notion of clique. Thus, a set of relaxed clique models has been proposed to overcome this issue \cite{cliquerelaxation13}. One of an interesting relaxed clique models is the $k$-plex which was first introduced in \cite{seidman1978graph}. In particular, given a graph $G$, a vertex set $C$ is a $k$-plex if every vertex has a degree at least $|C|-k$ within the subgraph of $G$ induced by $C$. A $k$-plex $C$ is said to be maximal if there does not exist any other $k$-plex that contains $C$. The problem of enumerating all maximal $k$-plexes from a graph has been widely studied in the literature \cite{07pakddWuP,BerlowitzCK15,WangCHSLPI17,ConteMSGMV18,ZhouXGXJ20}, which often arises in a large number of real-world applications, such as community detection in social networks \cite{11iorBalasundaramBH}, identifying protein complexes in protein-protein interaction networks \cite{bmcbiLuo09}, and serving as an alternative to cliques in biochemistry \cite{doyle2005robust}.
	
	However, the number of $k$-plexes in a real-world graph is often exponential to the size of the graph. It is well-known that the problem of enumerating all maximal $k$-plexes is NP-hard \cite{11iorBalasundaramBH}, resulting in that many existing algorithms for computing maximal $k$-plexes can only deal with very small graphs as reported in \cite{17kddConteFMPT,ConteMSGMV18}. Since small-size $k$-plexes are often no practical use in real-world applications, it is more desirable to compute all relatively-large maximal $k$-plexes with size no less than a given parameter $q$ (e.g., $q> 10$) \cite{17kddConteFMPT}. To enumerate all relatively-large maximal $k$-plexes, Conte et al.\ \cite{17kddConteFMPT} first developed an efficient algorithm using clique and $k$-core \cite{83kcoredef} to reduce the search space. Then, Conte et al.\  \cite{ConteMSGMV18} further proposed an improved algorithm with an effective pivoting technique which was originally used in many maximal clique enumeration algorithms \cite{cacmBronK73,13EppsteinLS}. However, the worst-case time complexity of these two algorithms is $O(n^2 2^n)$. More recently, a theoretically faster branch-and-bound algorithm was proposed by Zhou et al.\ \cite{ZhouXGXJ20}, which was the first algorithm with the worst-case time complexity lower than $O(n^2 2^n)$. {\color{black}Nevertheless, when setting $k\ge3$, all existing algorithms often require several hours to enumerate all relatively-large maximal $k$-plexes on middle-size real-world graphs. Moreover, to our knowledge, only the algorithm proposed in \cite{ConteMSGMV18} supports parallelism. However, as shown in our experiments, such a parallel algorithm is still very hard to handle large graphs even with 20 CPU cores. This motivates us to develop an efficient, parallel and scalable approach to enumerate all relatively-large maximal $k$-plexes on large real-world graphs.}
	%
	
	%
	%
	%
	%
	%
	%
	%
	
	\comment{The goal of this paper is to develop efficient solutions to improve the performance of enumerating all maximal $k$-plexes with size no less than $q$.}
	To achieve this goal, we first develop two novel upper-bounding techniques to reduce the search space, and then present a branch-and-bound enumeration algorithm based on a novel pivot re-selection technique. To further improve the scalability, we also devise a parallel version of our enumeration algorithm. In summary, we make the following contributions.
	
	\titleEmph{Novel upper bounds.} We propose two novel upper bounds for the $k$-plexes that contain a set $C$ of vertices. The key idea of our upper bounds is based on the principle that any possible maximal $k$-plex $C'$ containing $C$ can not break the rule that every vertex in $C$ must have at least $|C'|-k$ neighbors in $C'$. We show that with the help of our upper bounds, many unnecessary computations can be pruned in linear time.
	
	\titleEmph{Efficient algorithms.} We develop a new branch-and-bound algorithm to enumerate all maximal $k$-plexes with size no less than $q$. Specifically, in our algorithm, we first propose a novel pivot re-selection technique to reduce the search space, and then apply the proposed upper-bounding techniques to further prune unpromising branches. Theoretically, we prove that the worst-case time complexity of our algorithm is bounded by $O(n^2\gamma_k^n)$, where $\gamma_k$ is the branching factor of the algorithm with respect to $k$ which is strictly smaller than $2$. {\color{black}Finally, an efficient parallel algorithm is further designed to process large real-world graphs.}
	
	\titleEmph{Extensive experiments.} We conduct extensive experiments on real-world massive graphs to test the efficiency of our algorithm. The experiments show that our sequential algorithm can achieve up to $2\times$ to $100\times$ speedup over the state-of-the-art sequential algorithm on most benchmark graphs. {\color{black}The results also show that the speedup ratio of our parallel algorithm is almost linearly w.r.t.\ (with respect to) the number of threads, indicating the high scalability of the proposed parallel algorithm. For example, on the enwiki-2021 graph (with more than 300 million edges), when $k=2$ and $q=50$, our parallel algorithm only takes 71.2 seconds to enumerating all desired $k$-plexes using 20 threads, while the state-of-the-art parallel algorithm cannot terminate the computation within 24 hours under the same parameter settings.} \comment{For example, if $k=3$ and $q=20$, our algorithm can output all results within 27 seconds in Pokec (1,632,803 vertices and 22,301,964 edges), while the state-of-the-art algorithm takes around 821 seconds. 
	}For reproducibility purpose, the source code of this paper is released at \url{https://github.com/qq-dai/kplexEnum}.
	
	
	\section{Problem Statement} \label{sec:problemsatement}
	Given an undirected and unweighted graph $G=(V,E)$, where $V$ is the set of vertices and $E $ is the set of edges. Let $n = |V|$ and $m = |E|$ be the number of vertices and edges respectively. For each vertex $v$, the set of neighbors of $v$ in $G$, denoted by $N_v(G)$, is defined as $N_{v}(G) \triangleq \{u \in V|(v,u) \in E \}$. The degree of a vertex $v$ in $G$, denoted by $d_{v}(G)$, is the cardinality of $N_{v}(G)$, i.e., $d_{v}(G) = |N_{v}(G)|$. Further, we define the set of non-neighbors of $v \in V$ in $G$ as $\overline{N}_v(G)=V\setminus N_v(G)$ (Note that $v \in \overline{N}_v(G)$ for each $v \in V$). Similarly, the size of the set of non-neighbors of $v$ in $G$ is denoted by $\overline{d}_v(G)=|\overline{N}_v(G)|$. Let $G(S) = (S, E(S))$ be an induced subgraph of $G$ if $S \subseteq V$ and $E(S) = \{(u,v)|(u,v) \in E, u \in S, v \in S\} $. If the context is clear, we simply use $N_v(S)$ ($d_v(S)$) and $\overline{N}_v(S)$ ($\overline{d}_v(S)$) to denote $N_v(G(S))$ ($d_v(G(S))$) and $\overline{N}_v(G(S))$ ($\overline{d}_v(G(S))$), respectively. The definition of $k$-plex is given as follows.
	\begin{definition}[$k$-plex]
		\label{def:k-plex}
		Given an undirected graph $G$, a set $S$ of vertices is a $k$-plex if every vertex in the subgraph $G(S)$ induced by $S$ has a degree no less than $|S|-k$.
	\end{definition}
	A $k$-plex $S$ is said to be maximal if there is no other vertex set $S'\supset S$ in $G$ such that $S'$ is a $k$-plex. The problem of enumerating all maximal $k$-plexes has been widely studied in the literature \cite{07pakddWuP,BerlowitzCK15,WangCHSLPI17,ZhouXGXJ20}. Among them, \cite{ZhouXGXJ20} proposed the first algorithm with a nontrivial worst-case time guarantee based on a  branch-and-bound technique. This algorithm is also the state-of-the-art as far as we know. However, in real-world applications, when using $k$-plexes to detect communities, there may exist many maximal $k$-plexes with small size that are often no practical use \cite{ConteMSGMV18}. Therefore, it is more useful to enumerate the relatively-large maximal $k$-plexes for practical applications \cite{ConteMSGMV18}. Moreover, the relatively-large $k$-plexes are often much compact based on the following lemma.
	
	\begin{lemma}[\cite{seidman1978graph}]
		\label{lem:dimeter2}
		Given a  $k$-plex $S$ of $G$ with the size of $q$, the diameter of the subgraph induced by $S$ is at most $2$ if $q \ge 2k-1$.
	\end{lemma}
	
	By Lemma~\ref{lem:dimeter2}, it is easy to see that any maximal $k$-plex with size no less than $2k-1$ must be compactly connected and has a small diameter. Therefore, similar to \cite{ConteMSGMV18}, we aims to enumerate all maximal $k$-plexes with size no less than a given parameter $q\ge 2k-1$. More formally, we define our problem as follows.
	
	\stitle{Size-constraint maximal $k$-plex enumeration.} Given an undirected graph $G$, a positive integer $k$, and a size constraint $q \ge 2k-1$, the goal of our problem is to list all maximal $k$-plexes in $G$ with size no less than $q$.
	
	
	
	\section{The Pruning Techniques} \label{sec:upper-bounds}
	Before developing the algorithm, some useful properties of $k$-plex can be applied to improve the performance of maximal $k$-plex enumerations. Below, we first describe two existing pruning techniques, and then we develop several novel techniques to prune branches in $k$-plex enumeration.
	
	\subsection{Existing Pruning Techniques}
	To reduce the search space, a widely used concept of $k$-core can be used to filter the vertices of $G$ that are definitely not in the $k$-plexes with size of no less than $q$. In general, a $k$-core of $G$ is a maximal subgraph in which every vertex has a degree no less than $k$ within the subgraph \cite{83kcoredef}. Based on the definition of $k$-core, we can easily derive the following lemma. Due to the space limit, all proofs of this paper are given in the supplementary document.
	\begin{lemma}
		\label{lem:k-core-bound}
		All maximal $k$-plexes of $G$ with size no less than $q$ must be contained in the $(q-k)$-core of $G$.
	\end{lemma}

	\vspace*{-0.1cm}
	To compute the $(q-k)$-core of $G$, we can make use of a peeling algorithm developed in \cite{03omalgkcore}, which runs in $O(m+n)$ time. To further reduce unnecessary vertices, a more interesting property of $k$-plex is also discovered recently which has been successfully used in \cite{ConteMSGMV18,ZhouXGXJ20,ZhouHXF21}.
	
	\vspace*{-0.1cm}
	\begin{lemma}[\cite{ConteMSGMV18}]
		\label{lem:two-vertices-bound}
		Given a $k$-plex $S$ of $G$, for each pair of vertices $u$ and $v$ in $S$, we have:
		\begin{itemize}
			\item if $(u,v) \notin E(S)$, $|N_u(S)\cap N_v(S)| \ge |S|-2k+2$;
			\item if $(u,v) \in E(S)$, $|N_u(S)\cap N_v(S)| \ge |S|-2k$.
		\end{itemize}
	\end{lemma}
	
	In the enumeration procedure, we let $S$ be the $k$-plex to be extended in $G$, and let $C$ be the set of candidate vertices of $G$ that can be added to $S$ to form a larger $k$-plex. Then, given the size constraint $q$, we can use  Lemma~\ref{lem:two-vertices-bound} to prune unpromising candidate vertices in $C$. Specifically, for each vertex $u \in C$, we iteratively check whether there is a vertex $v \in S$ such that the vertices $u$ and $v$ conflict with Lemma~\ref{lem:two-vertices-bound} in the subgraph $G(S \cup C)$. If such a vertex $v \in S$ is found, it implies that the vertex $u$ cannot be added to $S$ to form a large $k$-plex. As a consequence, the vertex $u$ can be safely removed from $C$ without losing any result. Such an operation can be repeated until all vertices in $C$ meet the conditions given in Lemma~\ref{lem:two-vertices-bound}.
	
	As shown in \cite{ZhouHXF21}, the time complexity of removing vertices in $G$ using Lemma~\ref{lem:two-vertices-bound} is $O(lm^{1.5})$ by using a triangle listing algorithm proposed in \cite{Latapy08}, where $l$ is the number of iterations. Clearly, such a pruning technique is often very expensive when using in the recursive enumeration procedure. In the following, we will develop several novel and efficient pruning techniques to cut the branches in the recursive enumeration procedure.
	
	
	\subsection{Novel Upper-bounding Techniques}
	Here we develop several novel techniques to derive an upper bound of the size of a $k$-plex that contains some given vertices. Below, we start by giving a basic upper bound for the vertices in the candidate set $C$ in the enumeration procedure.
	\begin{lemma}
		\label{lem:nbr-bound}
		Given a $k$-plex $S$ and the candidate set $C$, for each $v \in C$, the upper bound of the $k$-plexes containing $S\cup \{v\}$ is $|S|+k-\overline{d}_v(S) + d_v(C)$, where $\overline{d}_v(S) = |S \setminus N_v(G)|$.
	\end{lemma}
	Clearly, by Lemma~\ref{lem:nbr-bound}, we can prune the vertex $v$ if the upper bound of the $k$-plexes containing $S\cup\{v\}$ is smaller than the given parameter $q$. However, such an upper bound is often very loose, as it is dependent on the degree $d_v(C)$. To improve this, we next develop two novel upper bounds of the vertices in the candidate set. The first solution is as follows. Let $N_v(C)=\{v_1, \cdots, v_d\}$ be the neighbors of $v$ in $C$ sorted in non-decreasing order of the size of $\overline{N}_{v_i}(S)$ where $v_i \in N_v(C)$, i.e., $|\overline{N}_{v_i}(S)| \le |\overline{N}_{v_{i+1}}(S)|$ for each $i\in [1, d)$. Let $T_v^i(C)$ be the first $i$ vertices in $N_v(C)$. Denote by $sup(S)=\sum_{v\in S}(k-|\overline{N}_v(S)|)$ the support number of non-neighbors of the $k$-plex $S$. Then, a tighter upper bound can be obtained according to the following lemma.
	\comment{
		\begin{figure}[t]
			\centering
			\includegraphics[width=0.4\linewidth]{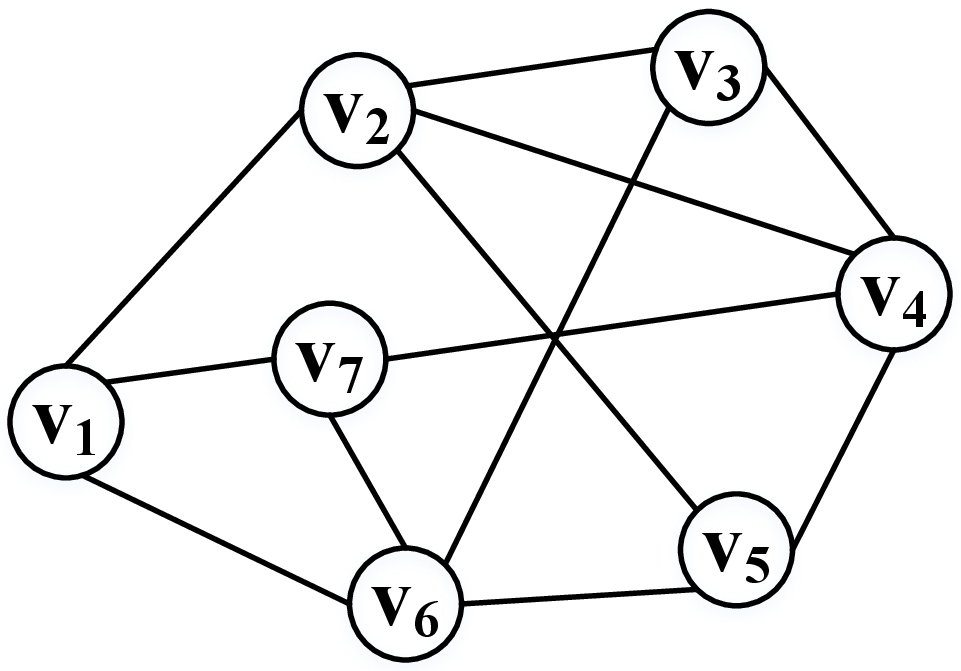}
			\vspace*{-0.2cm}
			\caption{A running example.}
			\vspace*{-0.2cm}
			\label{fig:example}
		\end{figure}
	}
	\begin{lemma}
		\label{lem:needed-bound}
		Given a $k$-plex $S$ and the candidate set $C$, let $\omega_k(S,C)$ be the maximum size of the $k$-plexes containing $S$ in $G(S \cup C)$. We have $\omega_k(S\cup\{v\},C) \le |S|+ k-\overline{d}_v(S) + max\{i|\sum_{u\in T_v^i(C)}|\overline{N}_u(S)| \le sup(S)\}$, where $v \in C$.
	\end{lemma}
	
	\begin{proof}
		Let $|S|+ k-\overline{d}_v(S)+r$ be the upper bound of the $k$-plexes containing $S\cup \{v\}$ with respect to Lemma~\ref{lem:needed-bound}. On the contrary, we assume that there is a $k$-plex $S'$ with the size of $|S|+ k-\overline{d}_v(S)+r'$ containing $S\cup \{v\}$, where $r' > r$. Thus, it is easily obtained that $S'$ contains at least $r'$ neighbors of $v$ in $C$. However, there are at most $r$ neighbors of $v$ in $C$ that can be added to $S$ with the restriction of the definition of $k$-plex. Thus, such assumption is contradiction, and the lemma is established.
	\end{proof}
	\comment{
		With respect to Lemma~\ref{lem:needed-bound}, we develop a novel pruning procedure described in Algorithm~\ref{alg:needed-nbr-prunning}. Specifically, the algorithm first makes use of an array $cnt$ to count the number of vertices in $N_v(C)$ that has a specific number of non-neighbors in the current $k$-plex $S$, i.e., $cnt(i)$ is the number of vertices in $N_v(C)$ who has $i$ non-neighbors in $S$ (lines~1-2). Then, the algorithm computes the support number $s=sup(S)$ of non-neighbors of the $k$-plex $S$ (line~3). Next, the algorithm traverses the array $cnt$ from the position $0$ to $k-1$. If the current $k$-plex can tolerate $cnt(i)$ vertices to push into $S$ with the constraint of $s$, the upper bound can be increased by $cnt(i)$. Otherwise, only $\lfloor s / i \rfloor$ vertices have a chance to push into $S$ (lines~5-9). After that, the algorithm can correctly return the upper bound of $k$-plexes containing $S\cup\{v\}$.
		
		It can be easily observed that the time complexity of Algorithm~\ref{alg:needed-nbr-prunning} is bounded by $O(k+|C|)\approx O(|C|)$ as $k$ is a small constant value.
		
		\begin{algorithm} [t]
			\small
			\caption{The support upper bound}
			\label{alg:needed-nbr-prunning}
			\KwIn{a $k$-plex $S$, the candidate set $C$, and $v \in C$.}
			\KwOut{The upper bound of $k$-plexes containing $S\cup \{v\}$.}
			\lFor{$i = 0$ to $k-1$}{$cnt(i) = 0$}
			\lFor{$ u \in N_v(C)$}{
				$cnt(\overline{d}_u(S))++$	
			}
			$s \gets \sum_{u\in S}(k-\overline{d}_u(S))$\;
			$ub \gets |S|+k-\overline{d}_v(S)$\;
			\For{$i = 0$ to $k-1$}{
				\If{$s - i*cnt(i) < 0$}{
					$ub \gets ub + \lfloor s / i \rfloor$; {\bf break}\;	
				}
				$ub \gets ub + cnt(i)$; $s \gets s - i*cnt(i)$\;
			}
			{\bf return $ub$}\;
		\end{algorithm}
	}
	\comment{
		It is easy to see that the upper bound given by Lemma~\ref{lem:needed-bound} is tighter than that given by Lemma~\ref{lem:nbr-bound}, and the following example further illustrates this point.
		\begin{example}
			Given an undirected graph $G$ shown in Fig.~\ref{fig:example} and the parameter $k=2$, let $S=\{v_1\}$ and $C=\{v_2,v_3,...,v_7\}$ be the current $k$-plex and the candidate set in a recursion, respectively. In this example, we assume that the vertex $v_2$ is used to expand $S$. By Lemma~\ref{lem:nbr-bound}, we can easily obtain that $\omega_k(S\cup\{v_2\},C)$ is bounded by $6$. However, when using Lemma~\ref{lem:needed-bound}, we have $\omega_k(S\cup\{v_2\},C) \le 5$. Because for the vertices in $N_{v_2}(C)=\{v_3,v_4,v_5\}$, at most two vertices can be added to $S\cup\{v_2\}$ to form a larger $k$-plex, which results a tighter bound on Lemma~\ref{lem:needed-bound} compared to Lemma~\ref{lem:nbr-bound}.
		\end{example}
	}
	
	Although Lemma~\ref{lem:needed-bound} can efficiently prune many unnecessary branches in the recursive enumeration procedure, we observe that it only considers the overall acceptable non-neighbors of vertices in the $k$-plex $S$. However, when applying Lemma~\ref{lem:needed-bound}, there may exist some vertices in $S$ that break the limit of only at most $k$ non-neighbors. For example, given a $k$-plex $S$, a vertex $v \in C$, and a neighbor set $N_v(C)=\{u_1,u_2,u_3\}$ of $v$, assume that $\overline{N}_{u_1}(S)=\overline{N}_{u_2}(S)=\overline{N}_{u_3}(S)=\{v'\}$ and $sup(S)$ equals to 3 w.r.t. $k=2$. According to Lemma~\ref{lem:needed-bound}, we cannot decrease the upper bound of the $k$-plexes containing $S\cup\{v\}$, while only one vertex in $N_v(C)$ may be added to $S\cup\{v\}$. Based on this analysis, we further develop an improved upper bound for each vertex in $C$. Let $sup(v, S)$ be the support number of non-neighbors of $v$ in $S$, i.e., $sup(v,S)=k-\overline{d}_v(S)$ where $v \in S$. Then, we have the following lemma.
	
	\begin{algorithm} [t]
		\small
		\caption{Compute the upper bound}
		\label{alg:hybrid-prunning}
		\KwIn{a $k$-plex $S$, the candidate set $C$, and a vertex $v \in C$.}
		\KwOut{The upper bound of $k$-plexes containing $S\cup \{v\}$.}
		\lFor{$i = 0$ to $k-1$}{$ B[i]= \emptyset$}
		\lForEach{$ u \in N_v(C)$}{
			$B[\overline{d}_u(S)].push(u)$	
		}
		\lForEach{$u\in S$}{$sup(u,S)=k-\overline{d}_u(S)$}
		$s \gets \sum_{u\in S}(k-\overline{d}_u(S))$\;
		$ub \gets |S|+k-\overline{d}_v(S)$\;
		\For{$i = 0$ to $k-1$}{
			\For{$u\in B[i]$ s.t. $s \ge i$}{
				$w \gets \arg \min\{sup(w,S)| w\in \overline{N}_u(S)\}$\;
				\If{$sup(w,S) > 0$}{
					$ub \gets ub + 1$;
					$s \gets s - i$\;
					$sup(w,S) \gets sup(w,S)-1$\;
				}
			}
		}
		{\bf return $ub$}\;
	\end{algorithm}
	
	\begin{lemma}
		\label{lem:vertex-support-condictions}
		Given a $k$-plex $S$, the candidate set $C$, and $sup(w, S)$ for each $w\in S$, for each $u$ in $N_v(C)$, where $v\in C$, we conduct the following operations: (1) get the vertex $w$ from $\overline{N}_u(S)$ whose support number of non-neighbors in $S$ is minimum, i.e., $sup(w,S)\le sup(w',S)$ for each $w'\in\overline{N}_u(S)$; (2) if $sup(w,S) > 0$, we decrease $sup(w,S)$ by 1; (3) otherwise, we remove $u$ from $N_v(C)$. Let $D$ be the remaining vertices in $N_v(C)$, then we have $\omega_k(S\cup\{v\},C) \le |S|+ k-\overline{d}_v(S) + |D|$.
	\end{lemma}
	\begin{proof}
		Suppose, on the contrary, that there is a $k$-plex $S'$ containing $S\cup\{v\}$ with $|S'| > |S|+ k-\overline{d}_v(S) + |D|$. Then, we can easily observe that there must be a set $D' \subseteq N_v(C)$ ($|D'| > |D|$) that can be added to $S\cup\{v\}$ to form a larger $k$-plex. Assume that $A=D' \setminus D$ and $B= D\setminus D'$. Then, we have that for each $u\in A$ there must exist a vertex $w \in \overline{N}_u(S)$ satisfying $sup(w,S) \le 0$ and $ w \in \overline{N}_{u'}(S)$ where $u'\in B$. Otherwise, the vertex $u$ can be added to $D$ according to Lemma~\ref{lem:vertex-support-condictions}. This indicates that the set $A$ must be the empty set, which is contradiction to the hypothesis. Thus, the lemma is established.
	\end{proof}
	
	\comment{
		The following example further clarifies the idea of Lemma~\ref{lem:vertex-support-condictions}.
		
		\begin{example}
			Given an undirected graph $G$ shown in Fig.~\ref{fig:example} and the parameter $k=3$, let $S=\{v_1,v_7\}$ and $C=\{v_2,v_3,...,v_6\}$ be the current $k$-plex and the candidate set of $S$ in a recursion, respectively. Suppose that the vertex $v_2$ will be added to $S$ in current recursion. The neighbor set  $\{v_3, v_4, v_5\}$ of $v_2$ in $C$ will be checked sequentially by Lemma~\ref{lem:vertex-support-condictions}. Before that, we first compute the support number of non-neighbors of each vertex in $S$ (i.e., $sup(v_1,S)=2, sup(v_7,S)=2$). Then, when checking the first vertex $v_3$, we only decrease $sup(v_1,S)$ by $1$, and for the second vertex $v_4$, $sup(v_1,S)$ also needs to be decreased, which leads to $sup(v_1,S)=0$. However, when checking the third vertex $v_5$, it would be removed from the neighbor set. Finally, only two vertices are left. Thus, we have $\omega_k(S\cup\{v_2\},C) \le 6$. Here if we make use of Lemma~\ref{lem:needed-bound} to generate the upper bound of $\omega_k(S\cup\{v_2\},C)$, we can obtain that $\omega_k(S\cup\{v_2\},C) \le 7$. This example demonstrates that Lemma~\ref{lem:vertex-support-condictions} works better than Lemma~\ref{lem:needed-bound}.
		\end{example}
	}
	We observe that a tighter upper bound can be further derived by combining Lemma~\ref{lem:needed-bound} with Lemma~\ref{lem:vertex-support-condictions}. The key idea is that when checking the $i$-th vertex in $N_v(C)$ whether it can be added to $S$ using Lemma~\ref{lem:needed-bound}, we also make use of Lemma~\ref{lem:vertex-support-condictions} to detect the vertex whether there is a conflict with the vertices in $\overline{N}_v(S)$. Based on this idea, we propose an upper-bounding algorithm which is shown in Algorithm~\ref{alg:hybrid-prunning}.
	
	In Algorithm~\ref{alg:hybrid-prunning}, it first makes use of a bucketing array $B$ to store the vertices in $N_v(C)$, where $R[i]$ contains all vertices in $N_v(C)$ that have $i$ non-neighbors in $S$ (lines~1-2). Then, the algorithm computes the support number of non-neighbors of each vertex in $S$ (line~3). After that, the algorithm sequentially checks each vertex in $B$ whether it can be added to the current $k$-plex $S$ based on Lemma~\ref{lem:needed-bound} and Lemma~\ref{lem:vertex-support-condictions} (lines~4-11). Specifically, when checking the vertex $v$ from $B[i]$ (line~7), the algorithm first selects $w$ from $\overline{N}_v(S)$ with the minimum support number of non-neighbors (line~8). If there is no conflict between $w$ and $u$ with respect to the definition of $k$-plex, the algorithm updates the corresponding values (lines~9-11). Otherwise, the vertex $u$ will not be considered. The algorithm continues the computations until all vertices have been examined. It is easy to derive that the time complexity of Algorithm~\ref{alg:hybrid-prunning} is bounded by $O(|S|+\text{log}k|C|) \approx O(|S|+|C|)$, as $k$ is often a very small constant for most real-world applications (e.g., $k\le 10$).

	\comment{	
		\begin{lemma}
			\label{lem:correctness-of-impalg}
			Algorithm~\ref{alg:hybrid-prunning} correctly returns the upper bound of the $k$-plexes containing $S \cup \{v\}$.
		\end{lemma}
		\begin{proof}
			This lemma is clearly established with respected to Lemma~\ref{lem:two-vertices-bound} and Lemma~\ref{lem:vertex-support-condictions}.
		\end{proof}
	}
	
	\section{A New Branch-and-Bound Algorithm} \label{sec:algorithms}
	
	\begin{algorithm} [t]
		\small
		\caption{The branch-and-bound algorithm.}
		\label{alg:enumeration}
		\SetKwProg{Fn}{Procedure:}{}{}
		\KwIn{The graph $G$ and two parameters $k$ and $q \ge 2k-1$.}
		\KwOut{All maximal $k$-plexes in $G$ with size no less than $q$.}
		
		Let $G'=(V',E')$ be the $(q-k)$-core subgraph of $G$; and let $\{v_1, ..., v_{n'}\}$ be the degeneracy ordering of $V'$\;
		\For{$i=1$ to $n'-q+1$}{
			$C \gets \{v_{i+1}, ..., v_{n'}\} \cap N_{v_i}^2(G')$\;
			$X \gets \{v_{1}, ..., v_{i-1}\} \cap N_{v_i}^2(G')$\;
			Pruning vertices in $C$ and $X$ using Lemma~\ref{lem:two-vertices-bound}\;
			$\bran(\{v_i\}, C, X, k, q)$\;
		}
		\Fn{$\bran(S, C, X, k, q)$}{
			\If{$C=\emptyset$}{
				\lIf{$|S|\ge q$ and $X=\emptyset$}{Output $S$}
				{\bf return}\;
			}
			Let $D$ be the subset of $S \cup C$ whose degree is minimum in $G(S\cup C)$\;
			Pick a pivot vertex $v$ from $D$ with the maximum $\overline{d}_v(S)$\;
			\If{$d_v(S\cup C) \ge |S\cup C| -k$}{
				\If{$S\cup C$ is a maximal $k$-plex in $G$}{
					\lIf{$|S\cup C|\ge q$}{Output $S\cup C$}
				}
				{\bf return}\;
			}
			\If{$v\in S$}{
				Re-pick a pivot vertex $u$ from $\overline{N}_v(C)$ using the same rules as used in lines~11-12; then set $v \gets u$\;
			}
			Compute the upper bound $ub$ for the pivot vertex $v$ by invoking Algorithm~\ref{alg:hybrid-prunning}\;
			\If{$ub \ge q$}{
				Get $C' \subset C$ and $X'\subseteq X$ such that $S\cup \{v,w\}$ is a $k$-plex for each $w \in C' \cup X'$\;
				$\bran(S \cup \{v\}, C', X', k, q)$\;	
			}
			$\bran(S, C\setminus\{v\}, X \cup \{v\}, k, q)$\;
		}		
	\end{algorithm}
	
	In this section, we present the detailed framework of our branch-and-bound enumeration algorithm. Let $S$ and $C$ be the current $k$-plex and the candidate set in a recursion, respectively. Our idea is that we first select a pivot vertex $v$ from $S\cup C$ according to the following rules: (i) the selected pivot vertex has the minimum degree in $G(S\cup C)$; and (ii) if there are multiple vertices that has the minimum degree in $G(S\cup C)$, we select the vertex among them that has maximum $\overline{d}_v(S)$. Then, we check whether the selected pivot vertex is contained in $S$. If so, we will re-choose a pivot vertex from the set $\overline{N}_v(C)$ using the same rules. Next, we make use of the (re-chosen) pivot vertex to conduct the branch-and-bound procedure. The detailed implementation of this algorithm is shown in Algorithm~\ref{alg:enumeration}.
	
	
	In Algorithm~\ref{alg:enumeration}, it first computes the $(q-k)$-core $G'$ of $G$, since all maximal $k$-plexes with size no less than $q$ are contained in $G'$ (Lemma~\ref{lem:dimeter2}). Then, the algorithm sorts the vertices in $V'$ with the degeneracy ordering $\{v_1, ..., v_{n'}\}$ (i.e., the vertex-removing ordering in the peeling algorithm for $k$-core decomposition \cite{03omalgkcore}) (line~1). Following the degeneracy ordering, the algorithm sequentially processes the vertices to enumerate all maximal $k$-plexes by invoking the procedure $\bran$, which admits five parameters: $S$, $C$, $X$, $k$ and $q$ (lines~2-6), where $S$, $C$, and $X$ are the three disjoint sets, respectively. In particular, $S$ is the current $k$-plex, $C$ is the candidate vertex set in which the vertex can be used to expand $S$, and $X$ is the set of excluded vertices that have already been explored in the previous recursions. Initially, $C$ ($X$) is set as the set of $2$-hop neighbors of $v_i$ (the set of vertices whose distance from $v_i$ is no larger than $2$, denoted by $N_{v_i}^2(G')$) that come after (before) $v_i$ (line~3-4), where $v_i$ is the $i$-th vertex in the degeneracy ordering. This is because the diameter of any desired maximal $k$-plex is no larger than 2 by Lemma~\ref{lem:dimeter2}. Then, the algorithm applies the results in Lemma~\ref{lem:two-vertices-bound} to prune vertices in $C$ and $X$ (line~5). After that, the algorithm invokes the recursive procedure $\bran$ to enumerate all maximal $k$-plexes containing $v_i$.
	
	In $\bran$, if $C \cup X$ is an empty set and the size of $S$ is no less than $q$, it outputs $S$ as a result (lines~8-10). Otherwise, it executes the branch-and-bound procedure (lines~11-24). In particular, the algorithm first selects a vertex $v$ with the maximum $\overline{d}_v(S)$ from the subset of $S\cup C$ that has minimum degree in $G(S\cup C)$ as a pivot vertex (lines~11-12). If such a pivot vertex is already in $S$ (i.e., $v \in S$), the algorithm re-chooses a pivot vertex from $\overline{N}_v(C)$ using the same rules (lines~17-18). Then, the algorithm computes the upper bound of $k$-plexes containing $S\cup \{v\}$ by using Algorithm~\ref{alg:hybrid-prunning} (line~19). If such an upper bound is no less than $q$, the algorithm recursively invokes the $\bran$ procedure by expanding $S$ with $v$ and updating $C$ ($X$) such that for each $w \in C'$ ($w \in X'$), $S\cup\{v,w\}$ is a $k$-plex (lines~20-22). After that, the algorithm performs the other recursive branch by moving $v$ from $C$ to $X$ (line~23). Finally, the algorithm terminates until the candidate set is empty or the minimum degree of $G(S\cup C)$ is no less than $|S \cup C|-k$ (line~8 or line~13).
	
	\subsection{Time Complexity Analysis} \label{subsec:time-complexity}
	The time complexity of Algorithm~\ref{alg:enumeration} is analyzed in Theorem~\ref{the:time-complexity}.
	\begin{theorem}
		\label{the:time-complexity}
		Given a graph $G$ and two parameters $k$ and $q$, Algorithm~\ref{alg:enumeration} returns all maximal $k$-plexes with  size no less than $q$ in time $O(n^2\gamma_k^n)$, where $\gamma_k < 2$ is the maximum positive real root of $x^{k+2}-2x^{k+1}+1=0$ (e.g. $\gamma_1 = 1.618$, $\gamma_2=1.839$, and $\gamma_3=1.928$).
	\end{theorem}
	\vspace*{-0.2cm}
	\begin{proof}
		The worst-case time complexity of Algorithm~\ref{alg:enumeration} is mainly dominated by the size of the recursive enumeration tree. Thus, in the following, we mainly focus on the bound of the number of enumeration branches in Algorithm~\ref{alg:enumeration}. Let $T(z, H)$ be the number of branches of $\bran(S, C, X, q,k)$, where $z$ is the size of $|C|$, and $H$ is the subgraph of $G$ induced by $S \cup C$. Denote by $H_i$ ($i\ge1$) the  induced subgraph of $H$ that a pivot vertex in $\bran(S\cup\{v_1, ..., v_{i-1}\}, C \setminus \{v_1, ..., v_{i-1}\}, X, q,k)$ is removed, where $\{v_1, ..., v_{i-1}\}= \emptyset$ if $i=1$. Then, we have the following recursive inequation:
		\vspace*{-0.1cm}
		\begin{equation}
			\vspace*{-0.05cm}
			\label{eqt:basic-recurrence}
			T(z, H) \le T(z-1, H)+T(z-1, H_1)
		\end{equation}

		Clearly, the branching factor of Eq.~(\ref{eqt:basic-recurrence}) is 2, which means that the time cost of Algorithm~\ref{alg:enumeration} is bounded by $O(P(n)2^n)$, where $P(n)$ is the worst-case running time of each branch in Algorithm~\ref{alg:enumeration}. However, a tighter bound can be obtained through the following analysis.
		
		\begin{enumerate}[leftmargin=*]
			\item[1.] If $C=\emptyset$ or $d_v(S\cup C) \ge |S\cup C|-k$, it is easy to derive that $T(z, H)$ is bounded by $1$. We then consider the other situations.
			\item[2.] If $v\in S$, a vertex $u$ from $\overline{N}_v(C)$ is chosen for branching. Interestingly, we observe that a pivot vertex $v$ in $\bran(S, C, X, q,k)$ is still the pivot vertex in the sub-branch $\bran(S\cup\{u\}, C\setminus\{u\}, X, q,k)$ in worst-case. Then, Eq.~(\ref{eqt:basic-recurrence}) can be further revised with the following recursive inequation:
			\vspace*{-0.1cm}
			\begin{equation}
				\vspace*{-0.1cm}
				\label{eqt:two-recurrence}
				T(z, H) \le T(z-1, H_1)+T(z-2, H_2) +T(z-2, H)
			\end{equation}
			We can perform similar substitutions on the sub-branches of $\bran(S\cup\{u\}, C\setminus\{u\}, X, q,k)$, until the final branch $\bran(S \cup P, C \setminus P, X, q,k)$ is obtained, where $|P|+\overline{d}_v(S)=k$.  Let $\overline{d}$ be the number of non-neighbors of $v$ in $C$. The maximum size of the candidate set of $S \cup P$ is $z-\overline{d}$, since only the neighbors of $v$ are left in this recursion. Finally, we have:
			\vspace*{-0.2cm}
			\begin{equation}
				\vspace*{-0.2cm}
				\label{eqt:S-recurrence}
				T(z, H) \le \sum_{i=1}^{p} T(z-i, H_i)+T(z-\overline{d}, H')
			\end{equation}
			where $p\le k-1$ and $\overline{d} > p$, and $H'$ the subgraph of $H$ induced by $S \cup C \setminus \overline{N}_v(C)$.
			
			\item[3.] If $v \notin S$, the vertex $v$ can be immediately used for branching. According to Eq.~(\ref{eqt:basic-recurrence}), the vertex $v$ can be the pivot vertex in $\bran(S\cup\{v\}, C\setminus\{v\}, X, q,k)$. Then, combining with Eq.~(\ref{eqt:S-recurrence}), we have:
			\vspace*{-0.2cm}
			\begin{equation}
				\vspace*{-0.2cm}
				\label{eqt:C-recurrence}
				T(z, H) \le \sum_{i=1}^{p+1} T(z-i, H_i)+T(z-\overline{d}-1, H')
			\end{equation}
			where $p\le k-1$ and $\overline{d} > p$.
		\end{enumerate}
		
		Note that in the worst case, $p$ and $\overline{d}$ are $k-1$ and $k$, respectively. Based on the theoretical result in \cite{10FominK}, the branching factor $r_k$ of $T(z, H)$ is the largest real root of function $x^{k+2}-2x^{k+1}+1=0$. As a result, the total time cost of Algorithm~\ref{alg:enumeration} can be bounded by $O(P(n)r_k^n)$. In each recursion, the time cost is dominated by the update operation (line~21) and the maximality checking operation (line~14), which are bounded by $O(n^2)$. Putting it all together, we have the time complexity of Algorithm~\ref{alg:enumeration}.
		
	\end{proof}

	\comment{
		
		The worst-case time complexity of Algorithm~\ref{alg:enumeration} is mainly dominated by the size of the recursive enumeration tree. In the following, we first discuss the bound of the number of enumeration branches in Algorithm~\ref{alg:enumeration}. Let $T(z, H)$ be the number of branches of $\bran(S, C, X, q,k)$, where $z$ is the size of $|C|$, and $H$ is the subgraph of $G$ induced by $S \cup C$. Denote by $H_i$ ($i\ge1$) the  induced subgraph of $H$ that a pivot vertex in $\bran(S\cup\{v_1, ..., v_{i-1}\}, C \setminus \{v_1, ..., v_{i-1}\}, X, q,k)$ is removed, where $\{v_1, ..., v_{i-1}\}= \emptyset$ if $i=1$. Then, we have the following recursive inequation:
		\vspace*{-0.05cm}
		\begin{equation}
			\vspace*{-0.05cm}
			\label{eqt:basic-recurrence}
			T(z, H) \le T(z-1, H)+T(z-1, H_1)
		\end{equation}

		Clearly, the branching factor of Eq.~(\ref{eqt:basic-recurrence}) is 2, which means that the time cost of Algorithm~\ref{alg:enumeration} is bounded by $O(P(n)2^n)$, where $P(n)$ is the worst-case running time of each branch in Algorithm~\ref{alg:enumeration}. However, a tighter bound can be obtained through the following analysis.
		
		\begin{enumerate}[leftmargin=*]
			\item[1.] If $C=\emptyset$ or $d_v(S\cup C) \ge |S\cup C|-k$, it is easy to derive that $T(z, H)$ is bounded by $1$. We then consider the other situations.
			\item[2.] If $v\in S$, a vertex $u$ from $\overline{N}_v(C)$ is chosen for branching. Interestingly, we observe that a pivot vertex $v$ in $\bran(S, C, X, q,k)$ is still the pivot vertex in the sub-branch $\bran(S\cup\{u\}, C\setminus\{u\}, X, q,k)$ in worst-case. Then, Eq.~(\ref{eqt:basic-recurrence}) can be further revised with the following recursive inequation:
			\vspace*{-0.1cm}
			\begin{equation}
				\vspace*{-0.1cm}
				\label{eqt:two-recurrence}
				T(z, H) \le T(z-1, H_1)+T(z-2, H_2) +T(z-2, H)
			\end{equation}
			We can perform similar substitutions on the sub-branches of $\bran(S\cup\{u\}, C\setminus\{u\}, X, q,k)$, until the final branch $\bran(S \cup P, C \setminus P, X, q,k)$ is obtained, where $|P|+\overline{d}_v(S)=k$.  Let $\overline{d}$ be the number of non-neighbors of $v$ in $C$. The maximum size of the candidate set of $S \cup P$ is $z-\overline{d}$, since only the neighbors of $v$ are left in this recursion. Finally, we have:
			\vspace*{-0.2cm}
			\begin{equation}
				\vspace*{-0.2cm}
				\label{eqt:S-recurrence}
				T(z, H) \le \sum_{i=1}^{p} T(z-i, H_i)+T(z-\overline{d}, H')
			\end{equation}
			where $p\le k-1$ and $\overline{d} > p$, and $H'$ the subgraph of $H$ induced by $S \cup C \setminus \overline{N}_v(C)$.
			
			\item[3.] If $v \notin S$, the vertex $v$ can be immediately used for branching. According to Eq.~(\ref{eqt:basic-recurrence}), the vertex $v$ can be the pivot vertex in $\bran(S\cup\{v\}, C\setminus\{v\}, X, q,k)$. Then, combining with Eq.~(\ref{eqt:S-recurrence}), we have:
			\vspace*{-0.2cm}
			\begin{equation}
				\vspace*{-0.2cm}
				\label{eqt:C-recurrence}
				T(z, H) \le \sum_{i=1}^{p+1} T(z-i, H_i)+T(z-\overline{d}-1, H')
			\end{equation}
			where $p\le k-1$ and $\overline{d} > p$.
		\end{enumerate}
		
		It is easy to show that in the worst case, $p$ and $\overline{d}$ are $k-1$ and $k$, respectively. Based on the theoretical result in \cite{10FominK}, we can derive that the branching factor $r_k$ of $T(z, H)$ is the largest real root of function $x^{k+2}-2x^{k+1}+1=0$. As a result, the total time cost of Algorithm~\ref{alg:enumeration} can be bounded by $O(P(n)r_k^n)$. In each recursion, the time cost is dominated by the update operation (line~21) and the maximality checking operation (line~14), which are bounded by $O(n^2)$. Putting it all together, we have the following result.
		\begin{theorem}
			\label{lem:time-complexity}
			Given a graph $G$ and two parameters $k$ and $q$, Algorithm~\ref{alg:enumeration} returns all maximal $k$-plexes with  size no less than $q$ in time $O(n^2\gamma_k^n)$, where $\gamma_k < 2$ is the maximum positive real root of $x^{k+2}-2x^{k+1}+1=0$ (e.g. $\gamma_1 = 1.618$, $\gamma_2=1.839$, and $\gamma_3=1.928$).
		\end{theorem}
	}
	
	\vspace*{-0.4cm}
	\stitle{Remark.} 
	As we analyzed before, the time complexity of Algorithm~\ref{alg:enumeration} is dominated by $T(z,H)$, which can be further determined by $p$ and $z-\overline{d}$ based on Eq.~(\ref{eqt:S-recurrence}) and Eq.~(\ref{eqt:C-recurrence}) when $k$ is given. More specifically, the smaller $p$ (or $z-\overline{d}$) yields a smaller branching factor of $T(z, H)$. In our pivot algorithm, the pivot vertex $v$ has a minimum degree in $G(S \cup C)$ and also has the maximum $\overline{d}_v(S)$ among all vertices who have the minimum degree in $G(S \cup C)$, which makes both $z-\overline{d}$ and $p$ as small as possible. 
	As a result, our pivot technique can achieve better practical performance compared to the existing pivot method \cite{ZhouXGXJ20}.

	\begin{table*}[t]
		\small
		\centering
		\caption{Running time of various sequential algorithms on 7 benchmark graphs.} \label{tab:exp-results}
		\vspace*{-0.3cm}
		\begin{tabular}{c|c|c|c|c|c|c|c|c|c|c|c|c|c}
			\toprule
			\mrdata{2} &   \mrk{2}   & \mrq{2} &\mrplex{2}& \mcltime{3}  & \mrdata{2} &  \mrk{2}  &  \mrq{2}  &\mrplex{2} & \mctime{3}    \\ \cline{5-7}\cline{12-14}
			&             &         &&  \fastPlex   & \commuPlex &  \dtwok   &            &             &    &&  \fastPlex   & \commuPlex &  \dtwok   \\ \hline
			\mrDBLP{7} & \mrow{2}{2} &   12    &12544&     0.12     &    2.3     & \bf{0.1}  & \mrSdot{7} & \mrow{3}{2} & 12 &27208777&  \bf{61.95}  &   220.24   &  184.86   \\ 
			&             &   20    &5049&     0.06     &    0.56    & \bf{0.05} &            &             & 20 &11411028&  \bf{25.39}  &   111.27   &  119.81   \\ \cline{2-7}
			& \mrow{2}{3} &   12    &3003588&  \bf{4.83}   &   17.66    &   7.77    &            &             & 30 &453&   \bf{0.2}   &    5.45    &   14.90   \\ \cline{9-14}
			&             &   20    &2141932&  \bf{3.33}   &   11.85    &   6.04    &            & \mrow{3}{3} & 12 &2807943240& \bf{6564.55} &  25721.92  &   19088.85   \\ \cline{2-7}
			& \mrow{3}{4} &   12    &610150817& \bf{727.06}  &  2929.28   &  1760.42  &            &             & 20 &1303148522& \bf{2641.46} &  14903.76  &   12045.24        \\ 
			&             &   20    &492253045&  \bf{576.8}  &   2489.0   &  1507.37  &            &             & 30 &1679468&  \bf{14.78}  &   141.95   &    1854.576       \\ \cline{9-14}
			&             &   30    &12088200&  \bf{21.82}  &   108.76   &  132.82   &            &      4      & 30 &502699966& \bf{1852.7}  &  26017.12  &   INF        \\ \hline
			\mrEml{6}  & \mrow{2}{2} &   12    &412779&  \bf{1.25}   &    4.61    &   9.23    & \mrWiki{6} & \mrow{2}{2} & 12 &2919931&  \bf{14.96}  &   41.11    &  101.09   \\ 
			&             &   20    &0&   \bf{0.1}   &    0.67    &   0.44    &            &             & 20 &52&  \bf{0.26}   &    2.24    &   10.25   \\ \cline{2-7}\cline{9-14}
			& \mrow{2}{3} &   12    &32639016&  \bf{82.24}  &   337.79   &  835.58   &            & \mrow{2}{3} & 12 &458153396& \bf{1948.39} &  6885.94   & 15510.07  \\ 
			&             &   20    &2637&  \bf{0.26}   &    3.24    &   48.14   &            &             & 20 &156727&  \bf{10.95}  &   149.58   &  1646.13  \\ \cline{2-7}\cline{9-14}
			& \mrow{2}{4} &   12    &1940182978& \bf{5345.09} &  31962.78  & 76634.4   &            & \mrow{2}{4} & 20 &46729532& \bf{421.86}  &  30153.31  &    INF    \\ 
			&             &   20    &1707177&  \bf{10.68}  &   307.62   &  3388.93  &            &             & 30 &0&  \bf{0.01}   &    2.24    &   0.05    \\ \hline
			\mrEps{5}  & \mrow{2}{2} &   12    &49823056& \bf{211.86}  &   475.79   &  624.26   & \mrPoc{8}  & \mrow{3}{2} & 12 &7679906&  \bf{48.27}  &  1040.74   &  167.02   \\ 
			&             &   20    &613&  \bf{13.44}  &   56.75    &  158.64   &            &             & 20 &94184&  \bf{7.19}   &   666.88   &   14.14   \\ \cline{2-7}
			& \mrow{2}{3} &   20    &548634119& \bf{1746.34} &  9637.53   &  28800.4  &            &             & 30 &3&     3.76     &   409.24   & \bf{2.78} \\ \cline{9-14}
			&             &   30    &16066&  \bf{3.37}   &   78.64    &  2172.89  &            & \mrow{3}{3} & 12 &520888893& \bf{1442.79} &  6861.47   &  12890.4  \\ \cline{2-7}
			&      4      &   30    &13172906& \bf{133.19}  &  20444.33  &    INF    &            &             & 20 &5911456&  \bf{27.43}  &   820.87   &  847.06   \\ \cline{1-7}
			\mrCai{3}  &      2      &   12    &5336&  \bf{0.04}   &    0.08    &   0.22    &            &             & 30 &5&     4.07     &   443.52   & \bf{3.87} \\ \cline{2-7}\cline{9-14}
			&      3      &   12    &281251&  \bf{0.74}   &    3.05    &   12.52   &            & \mrow{2}{4} & 20 &318035938& \bf{885.44}  &  13090.41  &    INF    \\ \cline{2-7}
			&      4      &   12    &15939883&  \bf{40.5}   &   185.16   &  724.91   &            &             & 30 &4515&  \bf{4.46}   &   489.5    &   82.98   \\ \bottomrule
		\end{tabular}
	\end{table*}
	
	\begin{figure*}[t] \vspace*{-0.5cm}
		\begin{center}
			\begin{tabular}[t]{c}
				\subfigure[{\scriptsize $k=2,q=10$}]{
					\includegraphics[width=0.45\columnwidth, height=2.5cm]{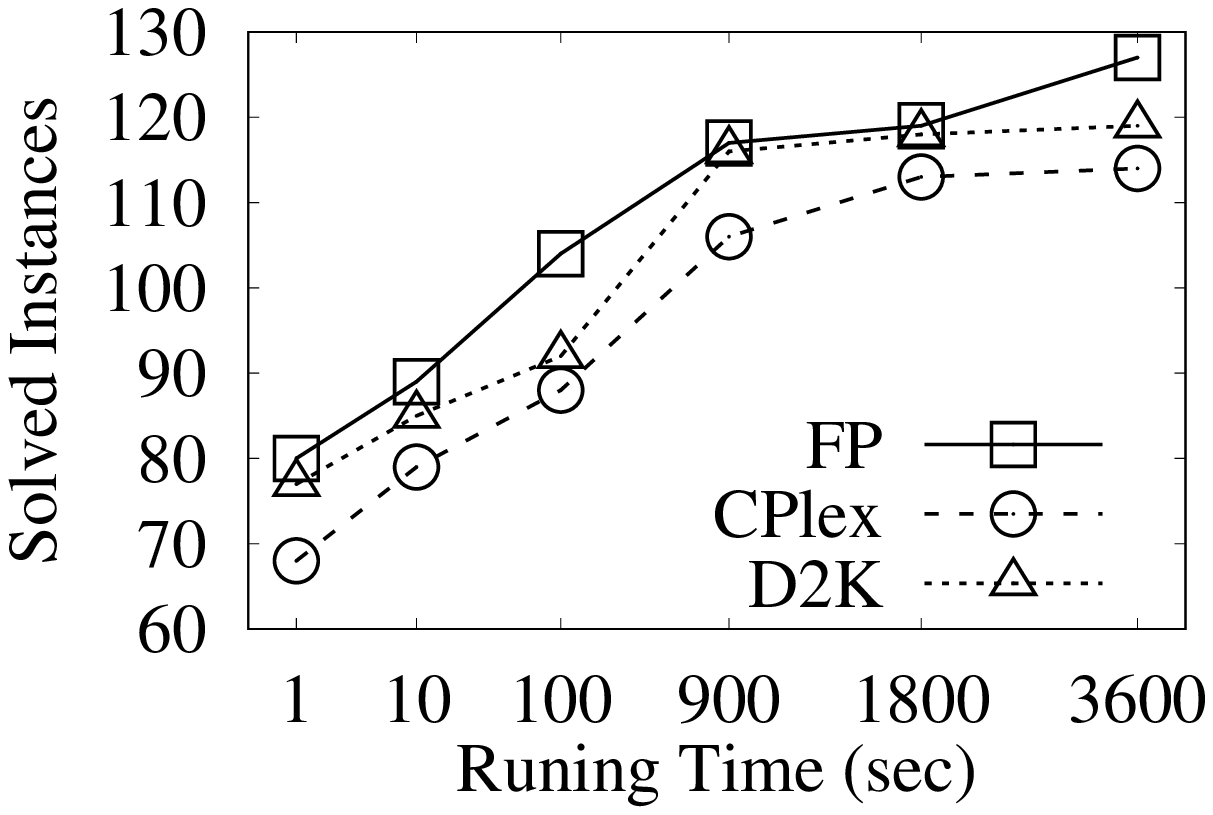}
				}
				\subfigure[{\scriptsize $k=3,q=10$}]{
					\includegraphics[width=0.45\columnwidth, height=2.5cm]{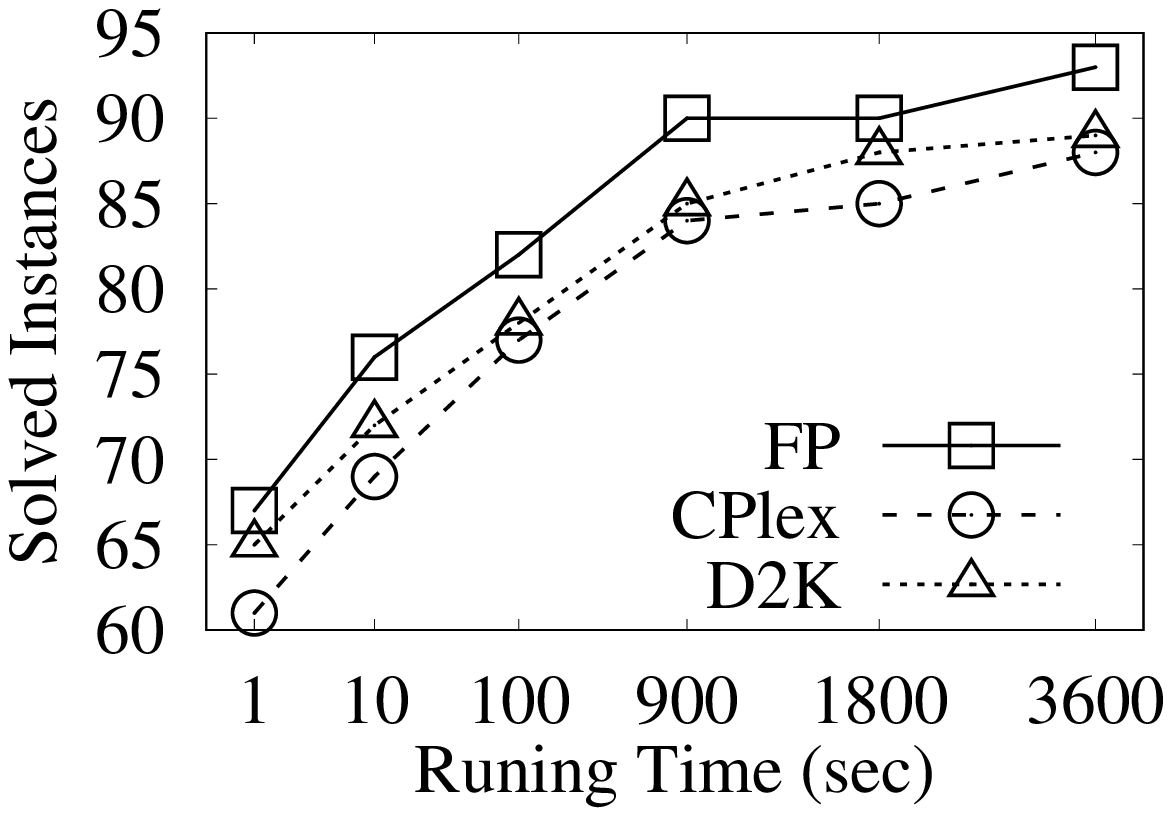}
				}
				\subfigure[{\scriptsize $k=4,q=10$}]{
					\includegraphics[width=0.45\columnwidth, height=2.5cm]{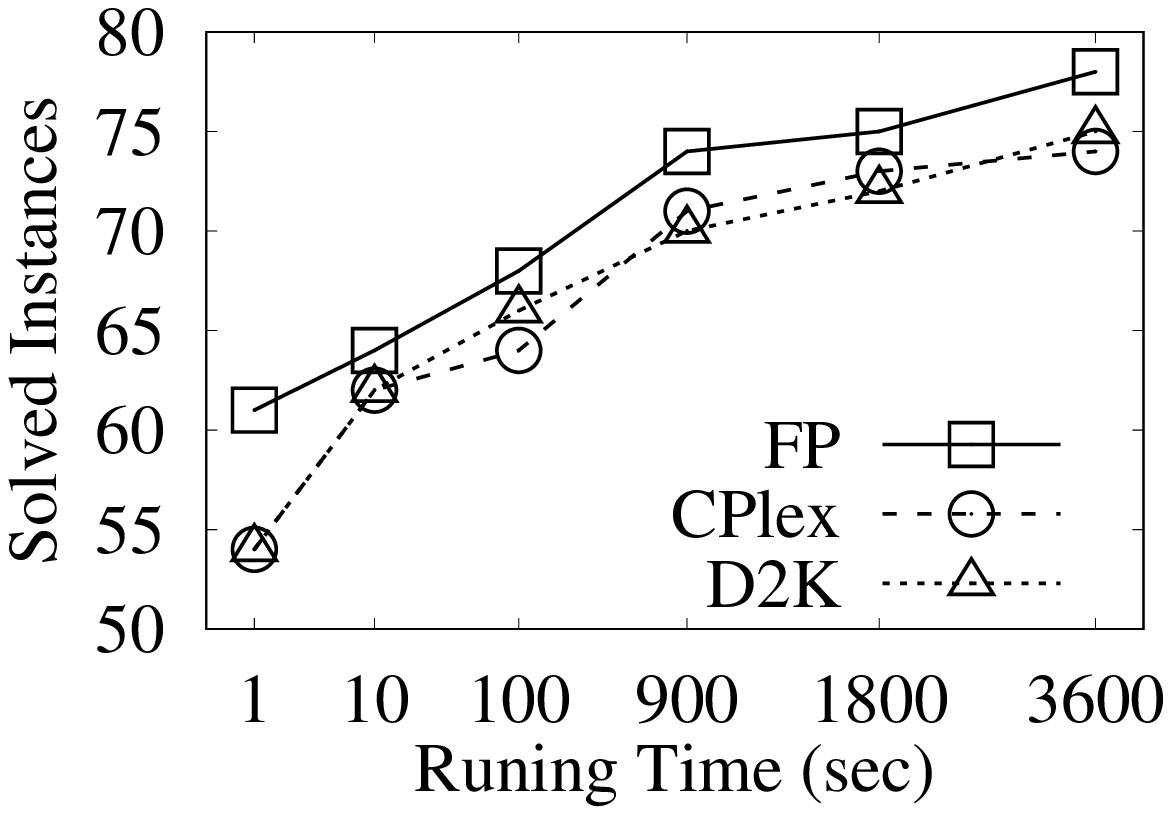}
				}	
				\subfigure[{\scriptsize $k=5,q=10$}]{
					\includegraphics[width=0.45\columnwidth, height=2.5cm]{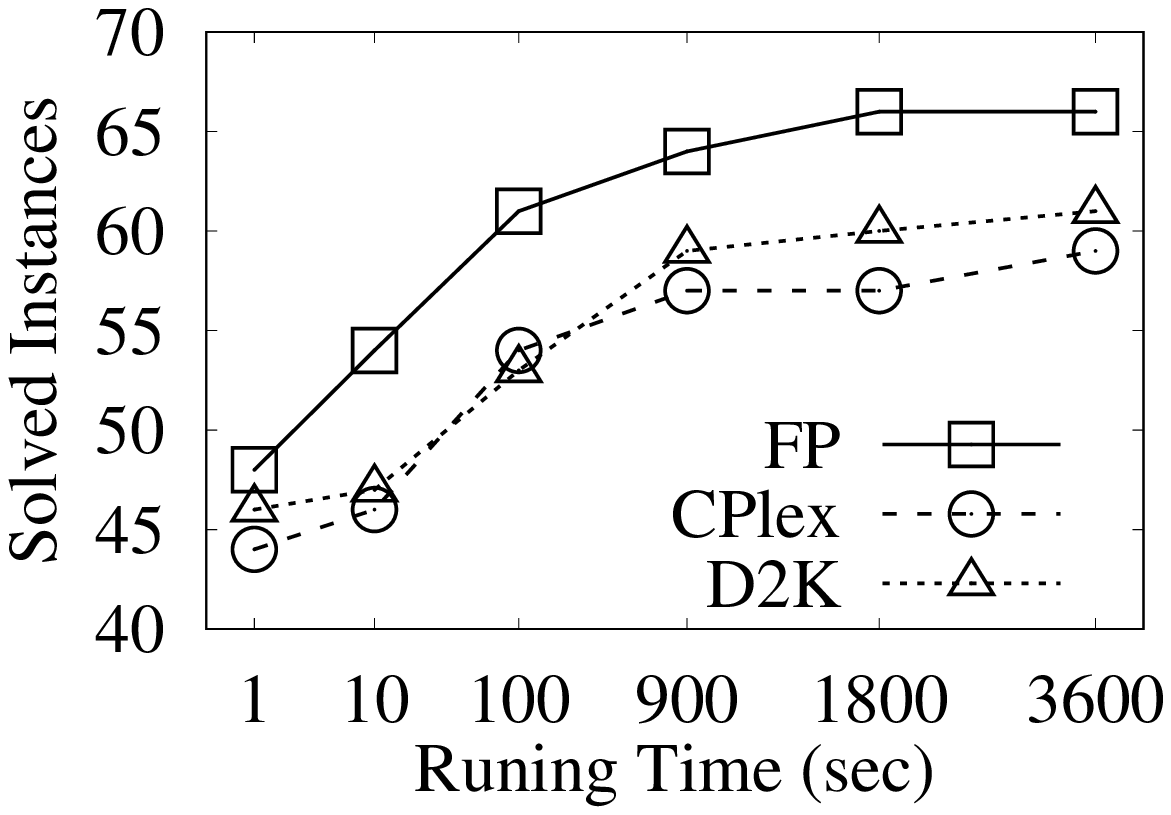}
				}\vspace*{-0.2cm} \\
				\subfigure[{\scriptsize $k=2,q=20$}]{
					\includegraphics[width=0.45\columnwidth, height=2.5cm]{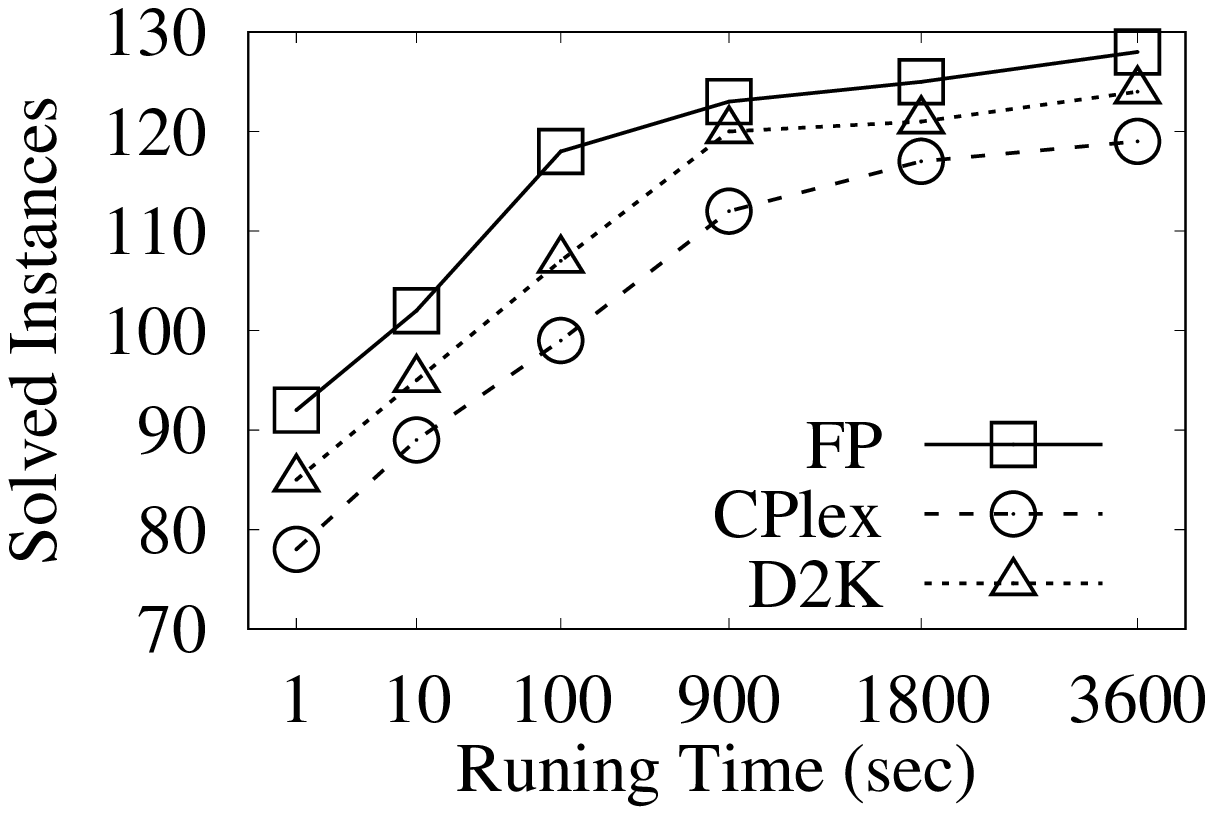}
				}
				\subfigure[{\scriptsize $k=3,q=20$}]{
					\includegraphics[width=0.45\columnwidth, height=2.5cm]{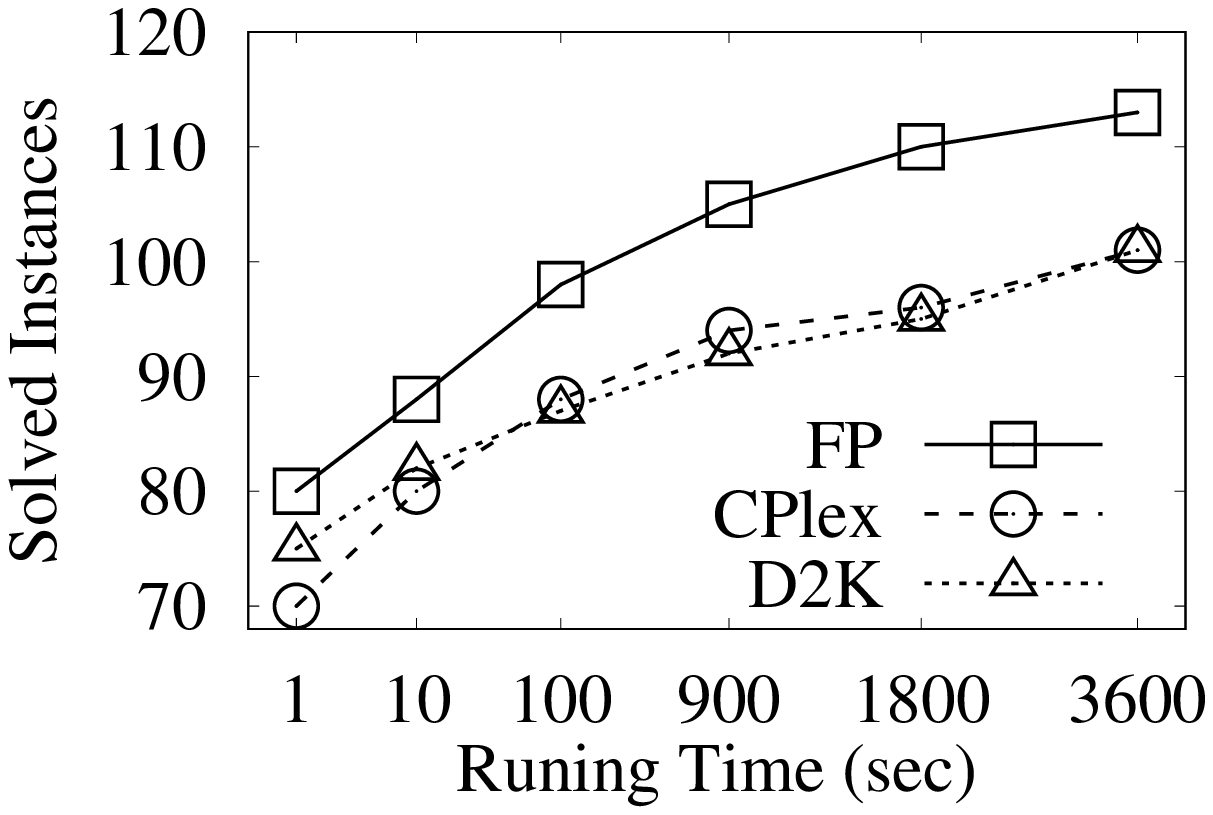}
				}
				\subfigure[{\scriptsize $k=4,q=20$}]{
					\includegraphics[width=0.45\columnwidth, height=2.5cm]{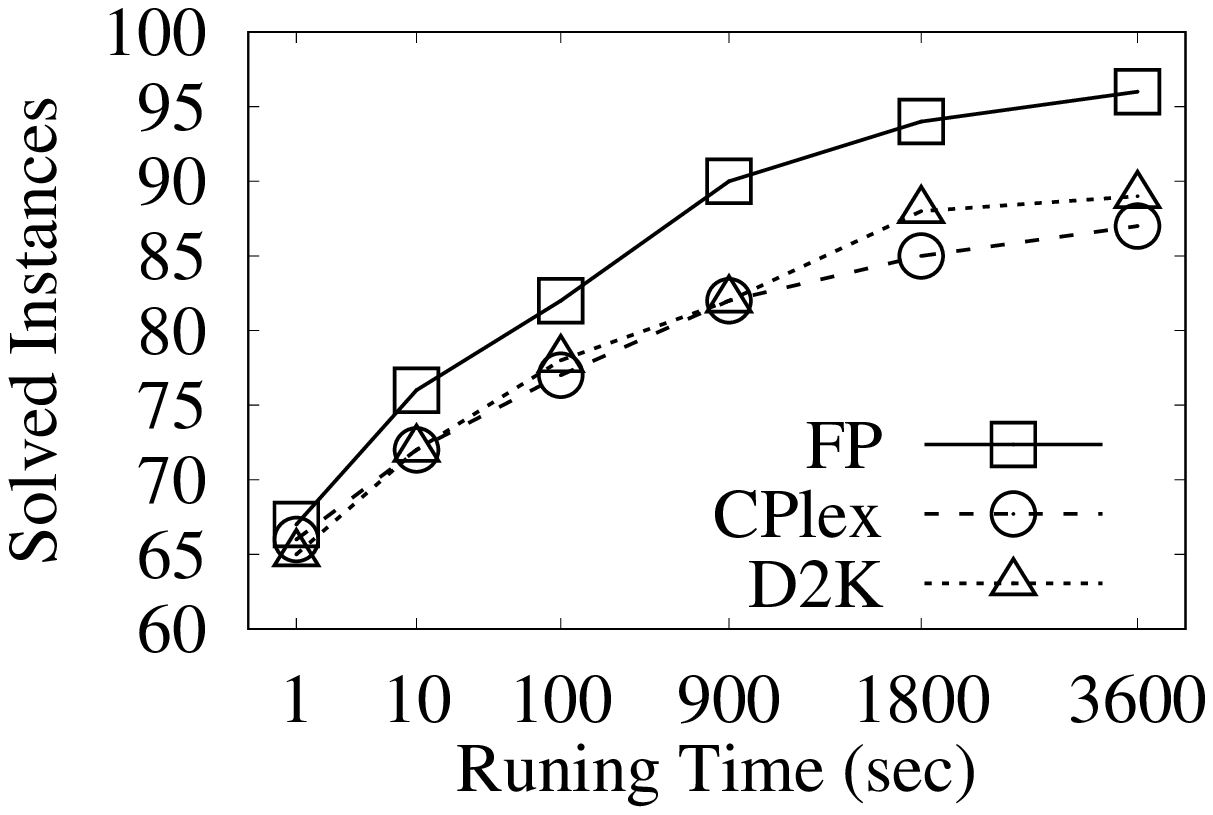}
				}	
				\subfigure[{\scriptsize $k=5,q=20$}]{
					\includegraphics[width=0.45\columnwidth, height=2.5cm]{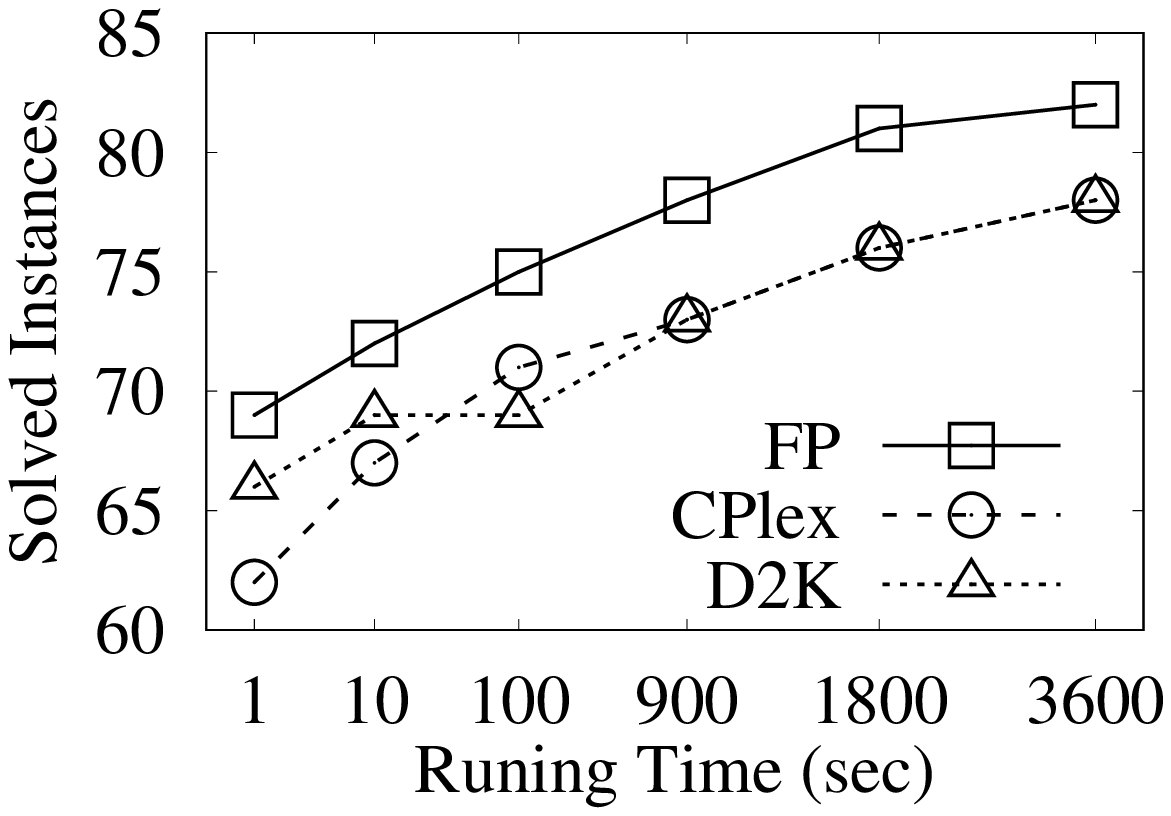}
				}
			\end{tabular}
		\end{center}
		\vspace*{-0.5cm}
		\caption{Number of solved instances on 139 real-world benchmark graphs under different time thresholds (in seconds).}
		\vspace*{-0.3cm}
		\label{fig:exp2:solved-instances}
	\end{figure*}
	
	\newcommand{\mrdatas}[1]{\multirow{#1}{*}{Datasets}}
	\newcommand{\mckl}[2]{\multicolumn{#1}{c|}{$k=#2$}}
	\newcommand{\mck}[2]{\multicolumn{#1}{c}{$k=#2$}}
	\newcommand{\mrbrct}[1]{\multirow{#1}{*}{\makecell[c]{brock200-2\\(200,19752)}}}
	\newcommand{\mrbrcf}[1]{\multirow{#1}{*}{\makecell[c]{brock200-4\\(200,26178)}}}
	\newcommand{\mrfattf}[1]{\multirow{#1}{*}{\makecell[c]{c-fat200-5\\(200,16946)}}}
	\newcommand{\mrfatff}[1]{\multirow{#1}{*}{\makecell[c]{c-fat500-5\\(500,46382)}}}
	\newcommand{\mrfatft}[1]{\multirow{#1}{*}{\makecell[c]{c-fat500-10\\(500,93254)}}}
	\newcommand{\mrjoh}[1]{\multirow{#1}{*}{\makecell[c]{johnson8-4-4\\(70,3710)}}}
	\newcommand{\mrmann}[1]{\multirow{#1}{*}{\makecell[c]{MANN\_a9\\(45,1836)}}}
	\newcommand{\mrhatt}[1]{\multirow{#1}{*}{\makecell[c]{p\_hat300-1\\(300,21866)}}}
	\newcommand{\mrhats}[1]{\multirow{#1}{*}{\makecell[c]{p\_hat700-1\\(700,121998)}}}

	\begin{table*}[h]
		\small
		\centering
		\caption{Running time of various sequential algorithms on DIMACS benchmarks (in second).}
		\label{tab:performance-on-DIMACS}
		\vspace*{-0.3cm}
		\begin{tabular}{c|c|c|c|c|c|c|c|c|c}
			\toprule
			\mrdata{2}  & \mrq{2} & \mckl{4}{2} &  \mck{4}{3}   \\   \cline{3-10}
			&         & \mrplex{1}  &   \fastPlex   & \commuPlex &  \dtwok   &\mrplex{1}&   \fastPlex   & \commuPlex &  \dtwok  \\ \hline
			\mrbrct{2}  &   10    &   2002262   &  \bf{22.24}   &   63.01    &  166.55   &2668680146& \bf{7030.18}  &  21218.15  & 41532.52 \\
			&   20    &      0      &   \bf{1.04}   &    7.93    &   41.92   &0			&  \bf{46.95}   &   787.89   & 10862.64 \\ \hline
			\mrbrcf{2}  &   10    & 5775998682  & \bf{17048.28} &  53907.15  & 36944.82  &	---		&      INF      &    INF     &   INF    \\
			&   20    &      2      &  \bf{215.63}  &  1097.01   &  9028.12  &19227523	& \bf{85958.46} &    INF     &   INF    \\ \hline
			\mrfattf{2} &   10    &    5721     &   \bf{0.13}   &    0.14    &   0.16    &1086435	&   \bf{5.28}   &   22.04    &  12.80   \\
			&   20    &    5721     &   \bf{0.09}   &    0.12    &   0.13    &1086435	&   \bf{3.66}   &   18.15    &  12.60   \\ \hline
			\mrfatff{2} &   10    &    15642    &     0.44      &    0.34    & \bf{0.23} &3576858	&  \bf{24.21}   &   44.97    &  26.94   \\
			&   20    &    15642    &      0.4      &    0.28    & \bf{0.23} &3576858	&  \bf{23.39}   &   41.24    &   25.8   \\ \hline
			\mrfatft{2} &   10    &    31258    &     3.91      &    5.1     & \bf{3.47} &29552680	&  \bf{574.5}   &  2585.94   & 1259.89  \\
			&   20    &    31258    &   \bf{2.61}   &    4.98    &   3.42    &29552680	&  \bf{524.66}  &  2533.09   & 1254.16  \\ \hline
			\mrjoh{2}  &   10    &   16047210  &  \bf{30.39}   &   76.71    &   63.14   &2019800107& \bf{5279.58}  &  17536.71  & 12888.40 \\
			&   20    &     0       &   \bf{0.17}   &    0.88    &   10.68   &0			&   \bf{26.1}   &   238.38   & 3650.27  \\ \hline
			\mrmann{2}  &   10    &   2160546   &   \bf{2.51}   &   10.25    &   6.75    &16619686	&  \bf{44.05}   &   395.74   &  75.39   \\
			&   20    &   1738656   &   \bf{2.05}   &    7.06    &   6.72    &16619686	&   \bf{42.8}   &   397.75   &  74.83   \\ \hline
			\mrhatt{2}  &   10    &     24      &   \bf{0.73}   &    3.3     &   7.97    &382654	&  \bf{52.27}   &   330.47   & 1022.41  \\
			&   20    &     0       &    \bf{0}     &    0.02    &   0.01    &0			&   \bf{0.22}   &    3.27    &  10.98   \\ \hline
			\mrhats{2}  &   10    &   422470    &  \bf{141.89}  &   545.85   &  1253.89  &3475000381& \bf{40311.47} &    INF     &   INF    \\
			&   20    &     0       &  \bf{15.94}   &   225.37   &  505.94   &0			& \bf{1343.98}  &  34751.76  &   INF    \\ \bottomrule
		\end{tabular}
	\end{table*}
	{\color{black}
		
		\subsection{Parallel Enumeration Strategy} \label{subsec:parallelization}
		Here we propose a simple but effective parallelization strategy for our enumeration algorithm when running on the multi-core machines. To achieve this, we need to split the computations into multiple independent sub-tasks. In Algorithm~\ref{alg:enumeration}, we can see that the branches that enumerate $k$-plexes containing the particular set $S$ are completely independent. So, a general parallel computation scheme is to divide the entire task into $n$ sub-tasks, and then we dynamically assign these sub-tasks to each thread to complete the computations (i.e., run lines 2-6 of Algorithm~\ref{alg:enumeration} in parallel). However, the computational workloads of some threads may be very large, leading to an inefficient parallel algorithm. The main reason for this is that the computational cost of each branch $\bran(S, C, X,q,k)$ is very different, which is mainly determined by the size of the set $C$ and the internal structural of the subgraph $G(S \cup C)$. Thus, in this paper, we develop an alternative solution to achieve a better load balancing. The details are as follows.
		
		In our implementation, we first make use of the above-mentioned general scheme to start the parallel computation. Then, during the branching calculations, the algorithm also monitors the work status of each thread. If the algorithm finds that a thread to be idle, some threads would divide the task being executed into two sub-tasks. The first sub-task is the sub-branch that computes the maximal $k$-plexes containing the pivot vertex, and the other sub-task is to compute the maximal $k$-plexes excluding the pivot vertex. Note that these two sub-tasks are also independent of each other. Thus, the newly generated sub-task can be safely assigned to the idle thread. This task division scheme is performed iteratively until all threads finished the computations. As shown in the experiments, such a parallel algorithm can achieve a very good speedup ratio over our sequential algorithm.
	}
	
	\section{Experiments} \label{sec:exp}
	In this section, we conduct extensive experiments to evaluate the efficiency of the proposed algorithm. Since this paper only focuses on improving the performance of maximal $k$-plex enumeration, we did not show the effectiveness testing for the maximal $k$-plex model. Below, we first describe the experimental setup and then report the results.
	
	\subsection{Experimental Setup}
	We implement our algorithm, called \fastPlex, in C++ to enumerate all maximal $k$-plexes, which combines the proposed upper bounding techniques and the pivot re-selection technique. For comparison, we use two state-of-the-art algorithms \dtwok \cite{ConteMSGMV18} and \commuPlex \cite{ZhouXGXJ20} as baselines, since all the other existing algorithms \cite{BerlowitzCK15,WangCHSLPI17,17kddConteFMPT} are less efficient than these two algorithms as shown in \cite{ConteMSGMV18,ZhouXGXJ20}. The C++ codes of  \dtwok and \commuPlex are provided by their authors, thus we use their original implementations in our experiments. All experiments are conducted on a PC with 2.2 GHz AMD CPU and 128GB memory running CentOS operating system.
	
	\comment{	
		\begin{table}[t!]
			\small
			\centering
			\caption{Number of $k$-plexes with varying parameters.} \label{tab:datas}
			\begin{tabular}{c|c|c|c|c}
				\hline
				Dataset & $q$ & $k=2$ & $k=3$ & $k=4$ \\ \hline
				dblp & 12 & 12544 & 3003588 & 610150817 \\ \hline
				dblp & 20 & 5049 & 2141932 & 492253045 \\ \hline
				EmaiEuall & 12 & 412779 & 32639016 & 1940182978 \\ \hline
				EmaiEuall & 20 & 0 & 2637 & 1707177 \\ \hline
				Epinions & 12 & 49823056 & -- & -- \\ \hline
				Epinions & 20 & 3322167 & 548634119 & -- \\ \hline
				Epinions & 30 & 0 & 16066 & 13172906 \\ \hline
				Caida & 12 & 5336 & 281251 & 15939883 \\ \hline
				Slashdot & 12 & 27208777 & 2807943240 & -- \\ \hline
				Slashdot & 20 & 11411028 & 1303148522 & -- \\ \hline
				Slashdot & 30 & 453 & 1679468 & 502699966 \\ \hline
				WikiVote & 12 & 2919931 & 458153396 & 9773156 \\ \hline
				WikiVote & 20 & 52 & 156727 & 46729532 \\ \hline
				Pokec & 12 & 7679906 & 520888893 & -- \\ \hline
				Pokec & 20 & 94184 & 5911456 & 318035938 \\ \hline
				Pokec & 30 & 3 & 5 & 4515 \\ \hline
			\end{tabular}
		\end{table}	
	}
	{\color{black}
		\stitle{Datasets.} In the experiments, we use three sets of graphs to evaluate the efficiency of the proposed algorithms. The first set of graphs is the real-world massive graphs containing 139 undirected and simple graphs collected from the Network Repository \cite{NetworkRepository}. The datasets can be downloaded from (\url{http://lcs.ios.ac.cn/~caisw/graphs.html}), and are widely used for evaluating the performance of $k$-plex search algorithms \cite{18ijcaimaxkplex,20AAAImaxkplex}. The second set of graphs is the DIMACS graphs (\url{http://archive.dimacs.rutgers.edu/pub/challenge/}), which are the well-known benchmark graphs for the test of maximal clique enumeration. The last set of graphs is the large real-world graphs, which are detailed in Table~\ref{tab:exp-parallel-large-graphs}.
	}
	
	\comment{
		\stitle{Datasets.} In the experiments, we use 139 undirected and simple real-world graphs in the Network Repository \cite{NetworkRepository} which include biological networks, social networks, technological networks, and so on. The datasets can be downloaded from (\url{http://lcs.ios.ac.cn/~caisw/graphs.html}). This set of benchmark graphs is widely used for evaluating the performance of $k$-plex search algorithms \cite{18ijcaimaxkplex,20AAAImaxkplex}.
	}
	
	{\color{black}
		\stitle{Parameters.} In all algorithms, there are two parameters: $k$ and the size constraint $q$. In the experiments, we select the parameters $k$ and $q$ from the interval $[2,5]$ and the $[10,30]$, respectively, as used in \cite{ZhouXGXJ20} for small or middle-size graphs. For large graphs, we adaptively set $q$ to find relatively-large $k$-plexes.
	}
	
	\subsection{Experimental Results} 
	\stitle{Exp-1: Efficiency of different sequential algorithms on 7 benchmark graphs.}
	\comment{First,} Here we compare the efficiency of different algorithms using 7 datasets selected from 139 benchmark graphs, because most of these 7 selected datasets have also been used as the benchmark datasets to evaluate different $k$-plex enumeration algorithms in \cite{ConteMSGMV18,ZhouXGXJ20}. Table~\ref{tab:exp-results} shows the running time of each algorithm with varying $k$ and $q$, where $k=2, 3, 4$ and $q=12, 20, 30$. If the algorithm can not terminate within 24 hours, we simply set its running time to ``INF''. From Table~\ref{tab:exp-results} , we can observe that our algorithm consistently outperforms \commuPlex on all datasets. In addition to a few results that are easy to be obtained by all algorithms, our algorithm is also much faster than \dtwok. In general, our algorithm can achieve $2\times$ to $100\times$ speedup over the state-of-the-art algorithms on most benchmark graphs. Moreover, the speedup of \fastPlex increases dramatically with the increase of $k$. This is because the proposed upper-bounding technique is the very effective in pruning unnecessary branches during the enumeration procedure. For instance, when $k=3$ and $q=30$, the speedups of \fastPlex over \commuPlex and \dtwok are $23\times$ and $644\times$ respectively on \socEps. On the same dataset, when $k=4$ and $q=30$, the speedup of \fastPlex over \commuPlex increases to $153\times$, and \dtwok even cannot finish the computation within 24 hours. These results demonstrate the high efficiency of the proposed algorithm.
	
	
	\stitle{Exp-2: Efficiency of different sequential algorithms on real-world graphs.}
	\comment{Second,} We test the number of solved instances of each algorithm on 139 real-world benchmark graphs to further compare the performance of different algorithms. Fig.~\ref{fig:exp2:solved-instances} shows the experimental results under different time thresholds with varying $k$ and $q$. As can be seen, \fastPlex solves the most number of instances among all algorithms with all parameter settings. When comparing with \commuPlex and \dtwok, we observe that \dtwok is usually superior to \commuPlex when $k \le 3$, because \dtwok is tailored for processing sparse real-world graphs. This result is consistent with the result shown in Table~\ref{tab:exp-results}. However, we can see that both \commuPlex and \dtwok still cannot keep up with our algorithm. This result further confirms that the proposed algorithm is very efficient to process real-world graphs. In addition, when computing large maximal $k$-plexes ($q=20$), the gap in the number of solved instances between \fastPlex and the state-of-the-art algorithms becomes very large on most parameter settings. For example, when $q=20$ and $k=4$, \fastPlex solved 90 instances with a time limitation of 900 seconds, while both \commuPlex and \dtwok solved 82 instances. The reason behind it is that the proposed upper-bounding techniques can reduce a large number of unnecessary computations and the pivot re-selection technique can further reduce the size of the candidate set.
	
	\stitle{Exp-3: Efficiency of different sequential algorithms on DIMACS graphs.}
	Here we also evaluate the performance of various algorithms on DIMACS graphs.
	Table~\ref{tab:performance-on-DIMACS} shows the experimental results of \fastPlex, \commuPlex and \dtwok on 9 DIMACS benchmark graphs with $k=2, 3$ and $q=10,20$. From Table~\ref{tab:performance-on-DIMACS}, it is easy to see that our algorithm is much faster than \commuPlex and \dtwok, except for a few results that are easy to be obtained by all algorithms. With the increase of $k$ or $q$, the speedup of our algorithm over the state-of-the-art algorithms increases accordingly. For example, when $k=2$ and $q=20$, the speedups of \fastPlex over \commuPlex and \dtwok on brock200-2 are $7.6\times$ and $40\times$ respectively, while when $k=3$ and $q=20$, the speedup of \fastPlex over \commuPlex and \dtwok on the same dataset are $16.7\times$ and $230\times$ respectively. These results further demonstrate the high efficiency of the proposed algorithm.
	
	\newcommand{\mrEmls}[1]{\multirow{#1}{*}{EmEuAll}}
	\newcommand{\mrEpss}[1]{\multirow{#1}{*}{Epinions}}
	\newcommand{\mrSdots}[1]{\multirow{#1}{*}{Slashdot}}
	\newcommand{\mrWikis}[1]{\multirow{#1}{*}{WikiVote}}
	\newcommand{\mrPocs}[1]{\multirow{#1}{*}{Pokec}}
	
	\begin{table}[t!]
		\small
		\centering
		\caption{Running time of \FPBsc and \fastPlex (in seconds).} \label{tab:exp-ub-testing}
		\vspace*{-0.3cm}
		\begin{tabular}{c|c|c|c|c|c}
			\toprule
			\mrdatas{2}   & \mrq{2} & \mckl{2}{3} &  \mck{2}{4}   \\ \cline{3-6}
			&         &   \FPBsc    &  \fastPlex   & \FPBsc  &  \fastPlex   \\ \hline
			\mrEmls{2}   &   12    &    109.4    &  \bf{82.24}  & 12219.2 & \bf{5345.09} \\ 
			&   20    &    0.49     &  \bf{0.26}   &  50.51  &  \bf{10.68}  \\ \hline
			\mrEpss{3}   &   12    &   49621.4   & \bf{42414.6} &   INF   &     INF      \\ 
			&   20    &   2257.03   & \bf{1746.34} &   INF   &     INF      \\ 
			&   30    &    9.17     &  \bf{3.37}   & 1892.27 & \bf{133.19}  \\ \hline
			\mrWikis{2}   &   12    &   3060.24   & \bf{1948.39} &   INF   &     INF      \\ 
			&   20    &    31.48    &  \bf{10.95}  & 6315.26 & \bf{421.86}  \\ \hline
			\mrPocs{3}   &   12    &   2069.68   & \bf{1442.79} &   INF   &     INF      \\ 
			&   20    &    36.62    &  \bf{27.43}  & 2455.05 & \bf{885.44}  \\ 
			&   30    &    4.45     &  \bf{4.07}   &  5.23   &  \bf{4.46}   \\ \bottomrule
		\end{tabular}
	\end{table}
	
	\newcommand{\mrnn}[1]{\multirow{#1}{*}{$n$}}
	\newcommand{\mrmn}[1]{\multirow{#1}{*}{$m$}}
	\newcommand{\mrdn}[1]{\multirow{#1}{*}{$d$}}
	\newcommand{\mrde}[1]{\multirow{#1}{*}{$\delta$}}
	\begin{table*}
		\small
		\centering
		\caption{Running time of parallel algorithms on large graphs using 20 threads (in second).}
		\label{tab:exp-parallel-large-graphs}
		\vspace*{-0.3cm}
		\begin{tabular}{c|c|c|c|c|c|c|c|c|c|c|c|c}
			\toprule
			\mrdatas{2}  & \mrnn{2} & \mrmn{2}  & \mrdn{2} & \mrde{2} & \mckl{4}{2} & \mck{4}{3}   \\ \cline{6-13}
			&          &           &          &          &     $q$     & \mrplex{1}  & \fastPlex & \dtwok & $q$ & \mrplex{1} & \fastPlex &  \dtwok  \\ \hline
			hollywood    & 2180759  & 457971264 &  13107  &   1297    &     550     &   3697275  & {\bf 49.61} & 3209.48 &  600   &  57187773     & {\bf 321.05} &  45774.19 \\ \hline
			enwiki-2021   & 6261141  & 300249854 &  232410  &   178    &     50     &   360  & {\bf 71.20} & INF &  50   &  40997 &  {\bf 14911.83}  &   INF     \\ \hline
			soc-orkut   & 2997166  & 212698416 &  27466   &   230    &     40      & 17607659790 & {\bf 10468.78} &  82290.98    & 50  &   6276699468    &   {\bf 17088.40} & INF \\ \hline
			fb-A-anon  & 3097165  & 47334788  &   4915   &    74    &     10      &  351549646  &  {\bf 124.41}  & 257.86  & 20  & 594505927  &   {\bf 142.29} & 3456.51  \\ \bottomrule
		\end{tabular}
	\end{table*}
	
	\stitle{Exp-4: The effect of the proposed upper bound techniques.}
	\comment{Finally,} Here we conduct an ablation experiment to study the effect of the upper-bounding techniques used in our algorithm. Let \FPBsc be the proposed branch-and-bound algorithm without using our upper-bounding techniques (i.e., Lemma~\ref{lem:needed-bound} and Lemma~\ref{lem:vertex-support-condictions}). Table~\ref{tab:exp-ub-testing} depicts the running time of \FPBsc and \fastPlex on four datasets with varying $k$ and $q$. The results on the other datasets are consistent. From Table~\ref{tab:exp-ub-testing}, we can see that the running time of \FPBsc is consistently higher than that of \fastPlex on all datasets. This is because the proposed upper bounds are very effective and easy to compute. Moreover, with the increase of $q$, the speedup of \fastPlex compared to \FPBsc usually increases, indicating that our upper bounds have a stronger pruning performance for a larger $q$. When comparing the results shown in Table~~\ref{tab:exp-results}, we can see that \FPBsc is consistently faster than \commuPlex on all datasets and also faster than \dtwok on most datasets with most parameter settings. This result indicates that the proposed pivot re-selection technique is indeed very effective as analyzed in Section~\ref{subsec:time-complexity}.
	
	\begin{figure}[t!] \vspace*{-0.2cm}
		\begin{center}
			\begin{tabular}[t]{c}\hspace*{-0.2cm}
				\subfigure[{\scriptsize $k=2,q=10$}]{
					\includegraphics[width=0.45\columnwidth, height=2.5cm]{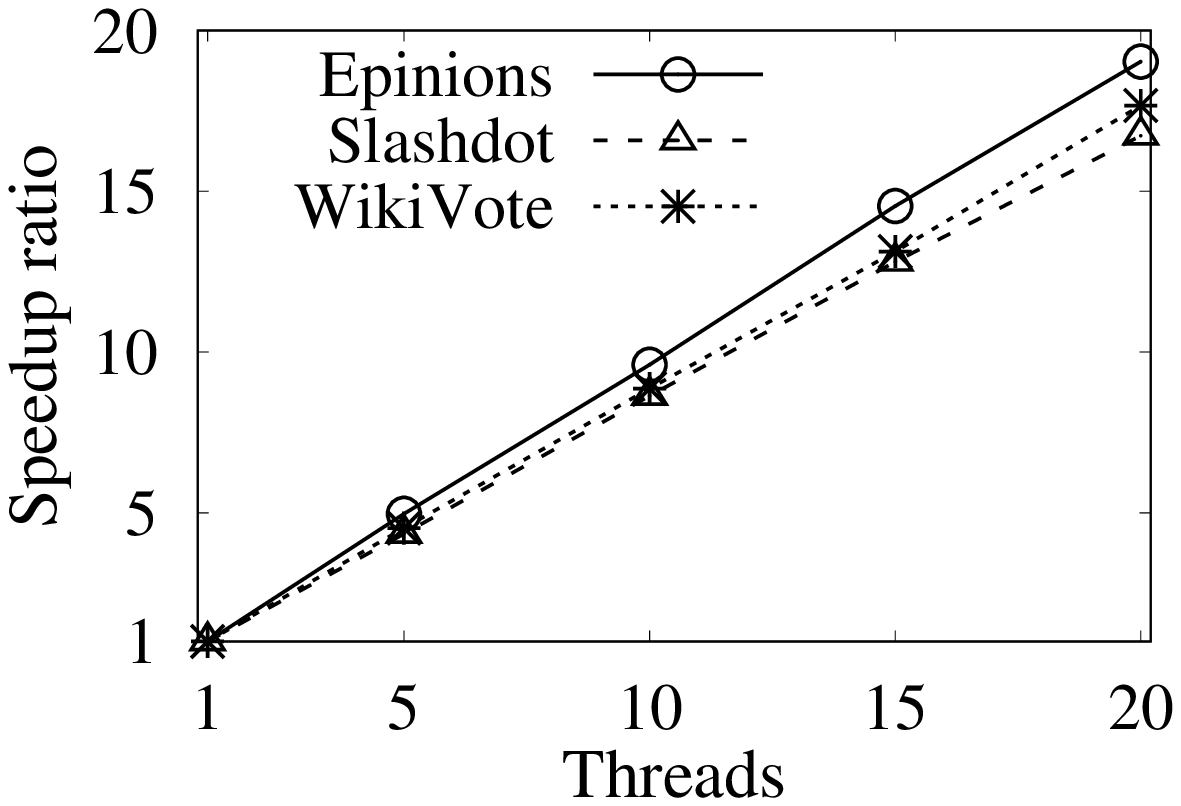}
				}\hspace*{-0.2cm}
				\subfigure[{\scriptsize $k=3,q=10$}]{
					\includegraphics[width=0.45\columnwidth, height=2.5cm]{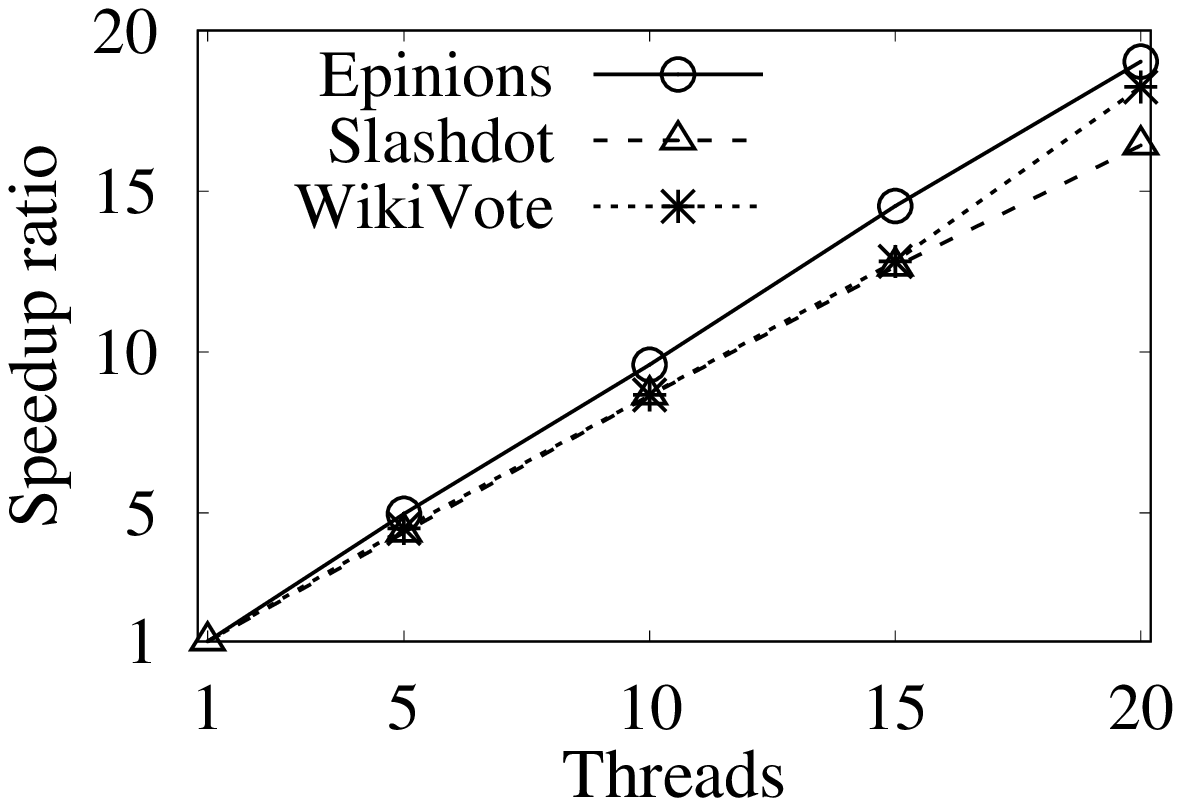}
				}\vspace*{-0.3cm}\\
				\hspace*{-0.2cm}
				\subfigure[{\scriptsize $k=3,q=20$}]{
					\includegraphics[width=0.45\columnwidth, height=2.5cm]{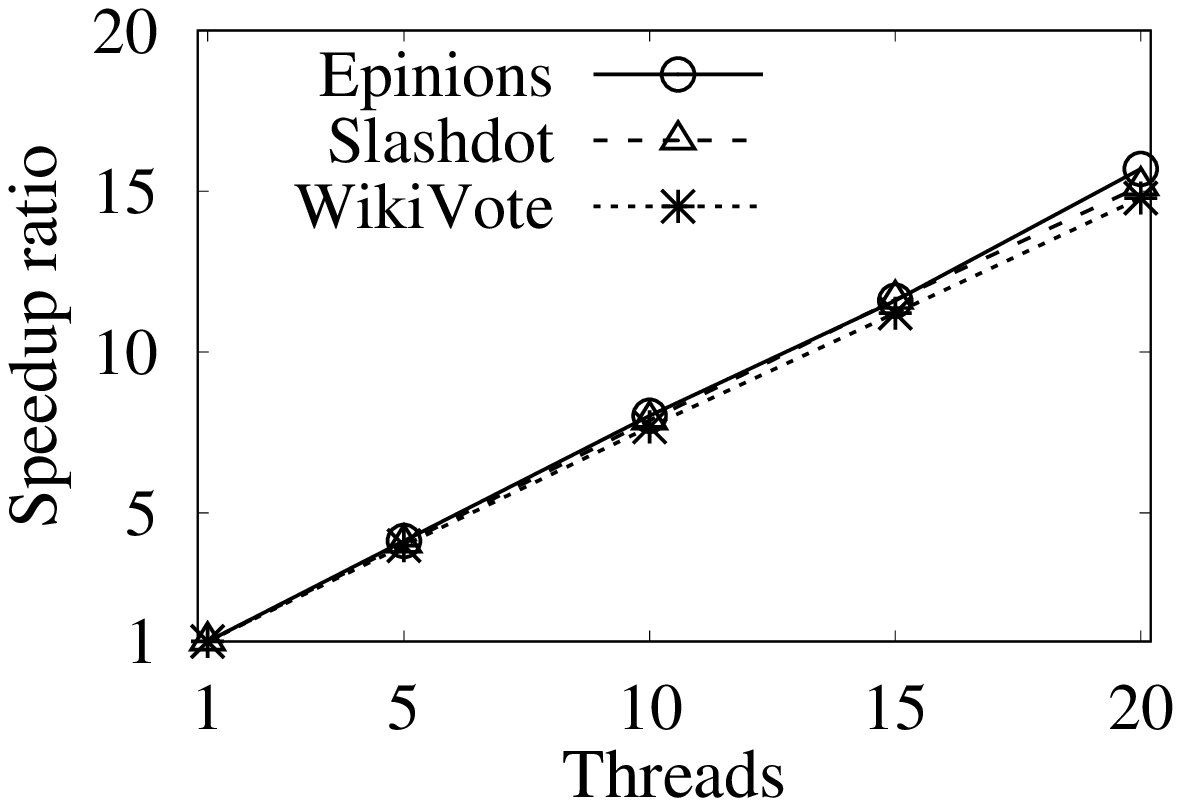}
				}\hspace*{-0.2cm}
				\subfigure[{\scriptsize $k=4,q=20$}]{
					\includegraphics[width=0.45\columnwidth, height=2.5cm]{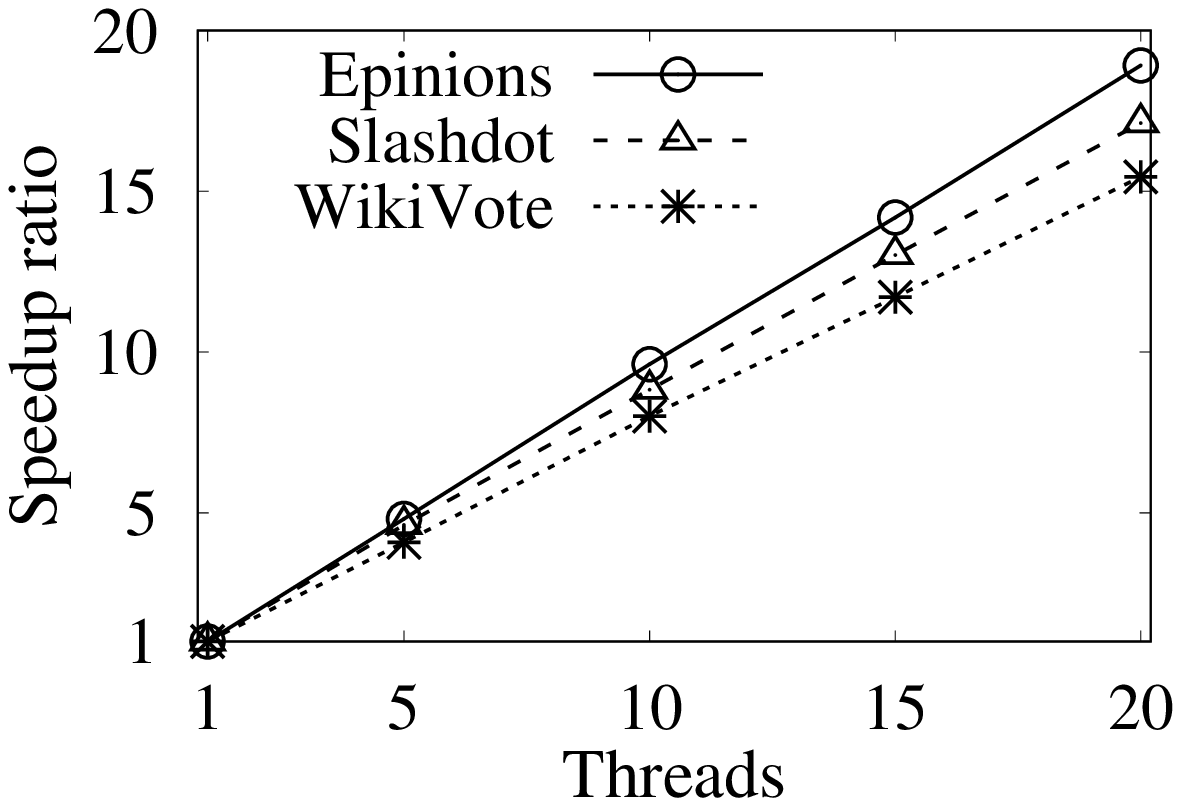}
				}
			\end{tabular}
		\end{center}
		\vspace*{-0.4cm}
		\caption{Speedup ratio of our parallel algorithm over the sequential algorithm.}
		\vspace*{-0.3cm}
		\label{fig:exp5:speedup-ratios}
	\end{figure}
	{\color{black}
		\stitle{Exp-5: The speedup ratio of our parallel algorithm.} In this experiment, we evaluate the speedup ratio of our parallel algorithm on benchmark graphs. Fig.~\ref{fig:exp5:speedup-ratios} shows the
		results on three benchmark datasets, and similar results can also be observed from the other datasets. As can be seen, the speedup ratio of our parallel algorithm is almost linear w.r.t.\ the number of threads used. More specifically, our parallel algorithm running with 20 threads is more than 15 times faster than the corresponding sequential algorithm (using one thread) with most parameter settings. For instance, when $k=2$ and $q=10$, on Epinions, Slashdot, and WikeVote, the parallel algorithm running with 20 threads is 19.03, 16.72, 17.66 times faster than the sequential algorithm, respectively. These results demonstrate the very high parallel performance of our algorithm.
		
		\stitle{Exp-6: Efficiency of the parallel algorithms on large real-world graphs.} Here we compare the performance of two parallel algorithms: our parallel algorithm and the parallel version of \dtwok algorithm, in processing large real-world graphs. Note that \commuPlex does not support the parallelism, thus we exclude such an algorithm in this experiment. The detailed statistics of the large graphs are shown in Table~\ref{tab:exp-parallel-large-graphs}, where columns $n$, $m$, $d$, and $\delta$ denote the number of vertices, edges, maximum degree and the degeneracy number of a graph, respectively. Each graph can be downloaded from the Network Repository \cite{NetworkRepository}  or \url{https://law.di.unimi.it/datasets.php}. Table~\ref{tab:exp-parallel-large-graphs} reports the performance of these two parallel algorithms with 20 threads. We can clearly see that our parallel algorithm substantially outperforms the state-of-the-art parallel algorithm. On some datasets, our parallel algorithm can achieve up to 100 times faster than the parallel version of \dtwok (e.g., on hollywood with $k=3$ and $q=600$), which further indicates the superiority of our parallel algorithm in processing large graphs.
	}
	
	
	{\color{black}
		\vspace*{-0.3cm}
		\section{Related Works} \label{sec:related-work}
		\stitle{Maximal clique enumeration.} Enumerating all maximal cliques from a graph is a fundamental problem in graph analysis. As shown in \cite{moon1965cliques}, the number of maximal cliques could be exponential w.r.t.\ the graph size , i.e., its number is $O(3^{n/3})$ in the worst case. However, the number of cliques in real-world graphs are often much smaller than such a worst-case bound. Many practical algorithms, including the classic Bron-Kerbosch algorithm \cite{cacmBronK73} and its pivot-based variants \cite{TomitaTT06worstcaseclique,13EppsteinLS,tcsNaude16}, work well on real-world graphs. \cite{edbtConte16,SegundoPrallelClique18} developed parallel pivot-based maximal clique enumeration algorithms to handle large datasets. There also exist many output-sensitive maximal clique enumeration algorithms \cite{MakinoU04,algorithmicaChangYQ13,icalpConteGMV16} which can achieve the polynomial-delay time complexity.
		
		\stitle{Maximal relaxed clique enumeration.} Since the constraint of clique is often very strictly for real-world community detection related applications, there exist many relaxed clique models including $k$-plex \cite{cliquerelaxation13}, $s$-defective cliques \cite{cliquerelaxation13}, $r$-clique \cite{edbtBeharC18}, $\gamma$-quasi-clique \cite{PeiJZ05,pkddLiuW08}, and so on. To enumerate all these relaxed cliques, many practical algorithms have been proposed in recent years. For example, \cite{edbtBeharC18} borrowed the solutions of enumerating all maximal clique to enumerate all $r$-cliques. \cite{jcssCohenKS08} developed a general framework to list all subgraphs with hereditary property which can also be used to enumerate $k$-plex and $s$-defective clique, as these two models satisfy the hereditary property. However, such a general framework is inefficient for enumerating a specific subgraph instance (e.g., $k$-plex or $s$-defective clique). Liu and Wong \cite{pkddLiuW08} presented a so-called quick algorithm to enumerate all $\gamma$-quasi-cliques. In this work, we focus on develop efficient, parallel, and scalable solution to enumerate all relatively-large $k$-plexes.
	}
	
	\vspace*{-0.3cm}	
	\section{Conclusion} \label{sec:conclusion}
	In this paper, we study the problem of enumerating all maximal $k$-plexes in a graph with size no less than a given parameter $q$. To solve this problem, we first develop two efficient upper bounds of the $k$-plexes containing a set $S$ of vertices. Then, we propose a branch-and-bound algorithm to enumerate maximal $k$-plexes based on a carefully-designed pivot re-selection technique and the proposed upper-bounding techniques. More importantly, we show that the worst-case time complexity of our algorithm is bounded by $O(n^2\gamma_k^n)$, where $\gamma_k$ is strictly smaller than 2. In addition, we also devise an effective parallelization strategy for the proposed enumeration algorithm. Extensive experimental results on more than 139 real-world graphs demonstrate the efficiency and scalability of the proposed techniques.

	
	
	\balance
	\bibliographystyle{ACM-Reference-Format}
	\bibliography{kplex}


\begin{thebibliography}{34}


\ifx \showCODEN    \undefined \def \showCODEN     #1{\unskip}     \fi
\ifx \showDOI      \undefined \def \showDOI       #1{#1}\fi
\ifx \showISBNx    \undefined \def \showISBNx     #1{\unskip}     \fi
\ifx \showISBNxiii \undefined \def \showISBNxiii  #1{\unskip}     \fi
\ifx \showISSN     \undefined \def \showISSN      #1{\unskip}     \fi
\ifx \showLCCN     \undefined \def \showLCCN      #1{\unskip}     \fi
\ifx \shownote     \undefined \def \shownote      #1{#1}          \fi
\ifx \showarticletitle \undefined \def \showarticletitle #1{#1}   \fi
\ifx \showURL      \undefined \def \showURL       {\relax}        \fi
\providecommand\bibfield[2]{#2}
\providecommand\bibinfo[2]{#2}
\providecommand\natexlab[1]{#1}
\providecommand\showeprint[2][]{arXiv:#2}

\bibitem[\protect\citeauthoryear{Balasundaram, Butenko, and Hicks}{Balasundaram
  et~al\mbox{.}}{2011}]%
        {11iorBalasundaramBH}
\bibfield{author}{\bibinfo{person}{Balabhaskar Balasundaram},
  \bibinfo{person}{Sergiy Butenko}, {and} \bibinfo{person}{Illya~V. Hicks}.}
  \bibinfo{year}{2011}\natexlab{}.
\newblock \showarticletitle{Clique Relaxations in Social Network Analysis: The
  Maximum \emph{k}-Plex Problem}.
\newblock \bibinfo{journal}{\emph{Oper. Res.}} \bibinfo{volume}{59},
  \bibinfo{number}{1} (\bibinfo{year}{2011}), \bibinfo{pages}{133--142}.
\newblock


\bibitem[\protect\citeauthoryear{Batagelj and Zaversnik}{Batagelj and
  Zaversnik}{2003}]%
        {03omalgkcore}
\bibfield{author}{\bibinfo{person}{Vladimir Batagelj} {and}
  \bibinfo{person}{Matjaz Zaversnik}.} \bibinfo{year}{2003}\natexlab{}.
\newblock \showarticletitle{An {O}(m) Algorithm for Cores Decomposition of
  Networks}.
\newblock \bibinfo{journal}{\emph{CoRR}}  \bibinfo{volume}{cs.DS/0310049}
  (\bibinfo{year}{2003}).
\newblock


\bibitem[\protect\citeauthoryear{Behar and Cohen}{Behar and Cohen}{2018}]%
        {edbtBeharC18}
\bibfield{author}{\bibinfo{person}{Rachel Behar} {and} \bibinfo{person}{Sara
  Cohen}.} \bibinfo{year}{2018}\natexlab{}.
\newblock \showarticletitle{Finding All Maximal Connected s-Cliques in Social
  Networks}. In \bibinfo{booktitle}{\emph{{EDBT}}}. \bibinfo{pages}{61--72}.
\newblock


\bibitem[\protect\citeauthoryear{Berlowitz, Cohen, and Kimelfeld}{Berlowitz
  et~al\mbox{.}}{2015}]%
        {BerlowitzCK15}
\bibfield{author}{\bibinfo{person}{Devora Berlowitz}, \bibinfo{person}{Sara
  Cohen}, {and} \bibinfo{person}{Benny Kimelfeld}.}
  \bibinfo{year}{2015}\natexlab{}.
\newblock \showarticletitle{Efficient Enumeration of Maximal k-Plexes}. In
  \bibinfo{booktitle}{\emph{{SIGMOD}}}. \bibinfo{pages}{431--444}.
\newblock


\bibitem[\protect\citeauthoryear{Bron and Kerbosch}{Bron and Kerbosch}{1973}]%
        {cacmBronK73}
\bibfield{author}{\bibinfo{person}{Coenraad Bron} {and} \bibinfo{person}{Joep
  Kerbosch}.} \bibinfo{year}{1973}\natexlab{}.
\newblock \showarticletitle{Finding All Cliques of an Undirected Graph
  (Algorithm 457)}.
\newblock \bibinfo{journal}{\emph{Commun. {ACM}}} \bibinfo{volume}{16},
  \bibinfo{number}{9} (\bibinfo{year}{1973}), \bibinfo{pages}{575--576}.
\newblock


\bibitem[\protect\citeauthoryear{Chang, Yu, and Qin}{Chang
  et~al\mbox{.}}{2013}]%
        {algorithmicaChangYQ13}
\bibfield{author}{\bibinfo{person}{Lijun Chang}, \bibinfo{person}{Jeffrey~Xu
  Yu}, {and} \bibinfo{person}{Lu Qin}.} \bibinfo{year}{2013}\natexlab{}.
\newblock \showarticletitle{Fast Maximal Cliques Enumeration in Sparse Graphs}.
\newblock \bibinfo{journal}{\emph{Algorithmica}} \bibinfo{volume}{66},
  \bibinfo{number}{1} (\bibinfo{year}{2013}), \bibinfo{pages}{173--186}.
\newblock


\bibitem[\protect\citeauthoryear{Chen, Wan, Cai, Li, and Chen}{Chen
  et~al\mbox{.}}{2020}]%
        {20AAAImaxkplex}
\bibfield{author}{\bibinfo{person}{Peilin Chen}, \bibinfo{person}{Hai Wan},
  \bibinfo{person}{Shaowei Cai}, \bibinfo{person}{Jia Li}, {and}
  \bibinfo{person}{Haicheng Chen}.} \bibinfo{year}{2020}\natexlab{}.
\newblock \showarticletitle{Local Search with Dynamic-Threshold Configuration
  Checking and Incremental Neighborhood Updating for Maximum k-plex Problem}.
  In \bibinfo{booktitle}{\emph{AAAI}}.
\newblock


\bibitem[\protect\citeauthoryear{Cohen, Kimelfeld, and Sagiv}{Cohen
  et~al\mbox{.}}{2008}]%
        {jcssCohenKS08}
\bibfield{author}{\bibinfo{person}{Sara Cohen}, \bibinfo{person}{Benny
  Kimelfeld}, {and} \bibinfo{person}{Yehoshua Sagiv}.}
  \bibinfo{year}{2008}\natexlab{}.
\newblock \showarticletitle{Generating all maximal induced subgraphs for
  hereditary and connected-hereditary graph properties}.
\newblock \bibinfo{journal}{\emph{J. Comput. Syst. Sci.}} \bibinfo{volume}{74},
  \bibinfo{number}{7} (\bibinfo{year}{2008}), \bibinfo{pages}{1147--1159}.
\newblock


\bibitem[\protect\citeauthoryear{Conte, Firmani, Mordente, Patrignani, and
  Torlone}{Conte et~al\mbox{.}}{2017}]%
        {17kddConteFMPT}
\bibfield{author}{\bibinfo{person}{Alessio Conte}, \bibinfo{person}{Donatella
  Firmani}, \bibinfo{person}{Caterina Mordente}, \bibinfo{person}{Maurizio
  Patrignani}, {and} \bibinfo{person}{Riccardo Torlone}.}
  \bibinfo{year}{2017}\natexlab{}.
\newblock \showarticletitle{Fast Enumeration of Large k-Plexes}. In
  \bibinfo{booktitle}{\emph{{KDD}}}. \bibinfo{pages}{115--124}.
\newblock


\bibitem[\protect\citeauthoryear{Conte, Grossi, Marino, and Versari}{Conte
  et~al\mbox{.}}{2016a}]%
        {icalpConteGMV16}
\bibfield{author}{\bibinfo{person}{Alessio Conte}, \bibinfo{person}{Roberto
  Grossi}, \bibinfo{person}{Andrea Marino}, {and} \bibinfo{person}{Luca
  Versari}.} \bibinfo{year}{2016}\natexlab{a}.
\newblock \showarticletitle{Sublinear-Space Bounded-Delay Enumeration for
  Massive Network Analytics: Maximal Cliques}. In
  \bibinfo{booktitle}{\emph{{ICALP}}} \emph{(\bibinfo{series}{LIPIcs})},
  Vol.~\bibinfo{volume}{55}. \bibinfo{pages}{148:1--148:15}.
\newblock


\bibitem[\protect\citeauthoryear{Conte, Matteis, Sensi, Grossi, Marino, and
  Versari}{Conte et~al\mbox{.}}{2018}]%
        {ConteMSGMV18}
\bibfield{author}{\bibinfo{person}{Alessio Conte}, \bibinfo{person}{Tiziano~De
  Matteis}, \bibinfo{person}{Daniele~De Sensi}, \bibinfo{person}{Roberto
  Grossi}, \bibinfo{person}{Andrea Marino}, {and} \bibinfo{person}{Luca
  Versari}.} \bibinfo{year}{2018}\natexlab{}.
\newblock \showarticletitle{{D2K:} Scalable Community Detection in Massive
  Networks via Small-Diameter k-Plexes}. In \bibinfo{booktitle}{\emph{{KDD}}}.
  \bibinfo{pages}{1272--1281}.
\newblock


\bibitem[\protect\citeauthoryear{Conte, Virgilio, Maccioni, Patrignani, and
  Torlone}{Conte et~al\mbox{.}}{2016b}]%
        {edbtConte16}
\bibfield{author}{\bibinfo{person}{Alessio Conte}, \bibinfo{person}{Roberto~De
  Virgilio}, \bibinfo{person}{Antonio Maccioni}, \bibinfo{person}{Maurizio
  Patrignani}, {and} \bibinfo{person}{Riccardo Torlone}.}
  \bibinfo{year}{2016}\natexlab{b}.
\newblock \showarticletitle{Finding All Maximal Cliques in Very Large Social
  Networks}. In \bibinfo{booktitle}{\emph{{EDBT}}}. \bibinfo{pages}{173--184}.
\newblock


\bibitem[\protect\citeauthoryear{Doyle, Alderson, Li, Low, Roughan, Shalunov,
  Tanaka, and Willinger}{Doyle et~al\mbox{.}}{2005}]%
        {doyle2005robust}
\bibfield{author}{\bibinfo{person}{John~C Doyle}, \bibinfo{person}{David~L
  Alderson}, \bibinfo{person}{Lun Li}, \bibinfo{person}{Steven Low},
  \bibinfo{person}{Matthew Roughan}, \bibinfo{person}{Stanislav Shalunov},
  \bibinfo{person}{Reiko Tanaka}, {and} \bibinfo{person}{Walter Willinger}.}
  \bibinfo{year}{2005}\natexlab{}.
\newblock \showarticletitle{The “robust yet fragile” nature of the
  Internet}.
\newblock \bibinfo{journal}{\emph{Proceedings of the National Academy of
  Sciences}} \bibinfo{volume}{102}, \bibinfo{number}{41}
  (\bibinfo{year}{2005}), \bibinfo{pages}{14497--14502}.
\newblock


\bibitem[\protect\citeauthoryear{Eppstein, L{\"{o}}ffler, and Strash}{Eppstein
  et~al\mbox{.}}{2013}]%
        {13EppsteinLS}
\bibfield{author}{\bibinfo{person}{David Eppstein}, \bibinfo{person}{Maarten
  L{\"{o}}ffler}, {and} \bibinfo{person}{Darren Strash}.}
  \bibinfo{year}{2013}\natexlab{}.
\newblock \showarticletitle{Listing All Maximal Cliques in Large Sparse
  Real-World Graphs}.
\newblock \bibinfo{journal}{\emph{{ACM} J. Exp. Algorithmics}}
  \bibinfo{volume}{18} (\bibinfo{year}{2013}).
\newblock


\bibitem[\protect\citeauthoryear{Fomin and Kratsch}{Fomin and Kratsch}{2010}]%
        {10FominK}
\bibfield{author}{\bibinfo{person}{Fedor~V. Fomin} {and}
  \bibinfo{person}{Dieter Kratsch}.} \bibinfo{year}{2010}\natexlab{}.
\newblock \bibinfo{booktitle}{\emph{Exact Exponential Algorithms}}.
\newblock


\bibitem[\protect\citeauthoryear{Gao, Chen, Yin, Chen, and Wang}{Gao
  et~al\mbox{.}}{2018}]%
        {18ijcaimaxkplex}
\bibfield{author}{\bibinfo{person}{Jian Gao}, \bibinfo{person}{Jiejiang Chen},
  \bibinfo{person}{Minghao Yin}, \bibinfo{person}{Rong Chen}, {and}
  \bibinfo{person}{Yiyuan Wang}.} \bibinfo{year}{2018}\natexlab{}.
\newblock \showarticletitle{An Exact Algorithm for Maximum k-Plexes in Massive
  Graphs}. In \bibinfo{booktitle}{\emph{IJCAI}}. \bibinfo{pages}{1449--1455}.
\newblock


\bibitem[\protect\citeauthoryear{Latapy}{Latapy}{2008}]%
        {Latapy08}
\bibfield{author}{\bibinfo{person}{Matthieu Latapy}.}
  \bibinfo{year}{2008}\natexlab{}.
\newblock \showarticletitle{Main-memory triangle computations for very large
  (sparse (power-law)) graphs}.
\newblock \bibinfo{journal}{\emph{Theor. Comput. Sci.}} \bibinfo{volume}{407},
  \bibinfo{number}{1-3} (\bibinfo{year}{2008}), \bibinfo{pages}{458--473}.
\newblock


\bibitem[\protect\citeauthoryear{Liu and Wong}{Liu and Wong}{2008}]%
        {pkddLiuW08}
\bibfield{author}{\bibinfo{person}{Guimei Liu} {and} \bibinfo{person}{Limsoon
  Wong}.} \bibinfo{year}{2008}\natexlab{}.
\newblock \showarticletitle{Effective Pruning Techniques for Mining
  Quasi-Cliques}. In \bibinfo{booktitle}{\emph{{PKDD}}},
  Vol.~\bibinfo{volume}{5212}. \bibinfo{pages}{33--49}.
\newblock


\bibitem[\protect\citeauthoryear{Luo, Li, Wan, and Scheuermann}{Luo
  et~al\mbox{.}}{2009}]%
        {bmcbiLuo09}
\bibfield{author}{\bibinfo{person}{Feng Luo}, \bibinfo{person}{Bo Li},
  \bibinfo{person}{Xiu{-}Feng Wan}, {and} \bibinfo{person}{Richard~H.
  Scheuermann}.} \bibinfo{year}{2009}\natexlab{}.
\newblock \showarticletitle{Core and periphery structures in protein
  interaction networks}.
\newblock \bibinfo{journal}{\emph{{BMC} Bioinform.}} \bibinfo{volume}{10},
  \bibinfo{number}{{S-4}} (\bibinfo{year}{2009}).
\newblock


\bibitem[\protect\citeauthoryear{Makino and Uno}{Makino and Uno}{2004}]%
        {MakinoU04}
\bibfield{author}{\bibinfo{person}{Kazuhisa Makino} {and}
  \bibinfo{person}{Takeaki Uno}.} \bibinfo{year}{2004}\natexlab{}.
\newblock \showarticletitle{New Algorithms for Enumerating All Maximal
  Cliques}. In \bibinfo{booktitle}{\emph{{SWAT}}}, Vol.~\bibinfo{volume}{3111}.
  \bibinfo{pages}{260--272}.
\newblock


\bibitem[\protect\citeauthoryear{Moon and Moser}{Moon and Moser}{1965}]%
        {moon1965cliques}
\bibfield{author}{\bibinfo{person}{John~W Moon} {and} \bibinfo{person}{Leo
  Moser}.} \bibinfo{year}{1965}\natexlab{}.
\newblock \showarticletitle{On cliques in graphs}.
\newblock \bibinfo{journal}{\emph{Israel journal of Mathematics}}
  \bibinfo{volume}{3}, \bibinfo{number}{1} (\bibinfo{year}{1965}),
  \bibinfo{pages}{23--28}.
\newblock


\bibitem[\protect\citeauthoryear{Naud{\'{e}}}{Naud{\'{e}}}{2016}]%
        {tcsNaude16}
\bibfield{author}{\bibinfo{person}{Kevin~A. Naud{\'{e}}}.}
  \bibinfo{year}{2016}\natexlab{}.
\newblock \showarticletitle{Refined pivot selection for maximal clique
  enumeration in graphs}.
\newblock \bibinfo{journal}{\emph{Theor. Comput. Sci.}}  \bibinfo{volume}{613}
  (\bibinfo{year}{2016}), \bibinfo{pages}{28--37}.
\newblock


\bibitem[\protect\citeauthoryear{Pattillo, Youssef, and Butenko}{Pattillo
  et~al\mbox{.}}{2013}]%
        {cliquerelaxation13}
\bibfield{author}{\bibinfo{person}{Jeffrey Pattillo}, \bibinfo{person}{Nataly
  Youssef}, {and} \bibinfo{person}{Sergiy Butenko}.}
  \bibinfo{year}{2013}\natexlab{}.
\newblock \showarticletitle{On clique relaxation models in network analysis}.
\newblock \bibinfo{journal}{\emph{Eur. J. Oper. Res.}} \bibinfo{volume}{226},
  \bibinfo{number}{1} (\bibinfo{year}{2013}), \bibinfo{pages}{9--18}.
\newblock


\bibitem[\protect\citeauthoryear{Pei, Jiang, and Zhang}{Pei
  et~al\mbox{.}}{2005}]%
        {PeiJZ05}
\bibfield{author}{\bibinfo{person}{Jian Pei}, \bibinfo{person}{Daxin Jiang},
  {and} \bibinfo{person}{Aidong Zhang}.} \bibinfo{year}{2005}\natexlab{}.
\newblock \showarticletitle{On mining cross-graph quasi-cliques}. In
  \bibinfo{booktitle}{\emph{{KDD}}}. \bibinfo{pages}{228--238}.
\newblock


\bibitem[\protect\citeauthoryear{Rossi and Ahmed}{Rossi and Ahmed}{2015}]%
        {NetworkRepository}
\bibfield{author}{\bibinfo{person}{Ryan~A. Rossi} {and}
  \bibinfo{person}{Nesreen~K. Ahmed}.} \bibinfo{year}{2015}\natexlab{}.
\newblock \showarticletitle{The Network Data Repository with Interactive Graph
  Analytics and Visualization}. In \bibinfo{booktitle}{\emph{AAAI}}.
\newblock
\urldef\tempurl%
\url{https://networkrepository.com}
\showURL{%
\tempurl}


\bibitem[\protect\citeauthoryear{Segundo, Artieda, and Strash}{Segundo
  et~al\mbox{.}}{2018}]%
        {SegundoPrallelClique18}
\bibfield{author}{\bibinfo{person}{Pablo~San Segundo}, \bibinfo{person}{Jorge
  Artieda}, {and} \bibinfo{person}{Darren Strash}.}
  \bibinfo{year}{2018}\natexlab{}.
\newblock \showarticletitle{Efficiently enumerating all maximal cliques with
  bit-parallelism}.
\newblock \bibinfo{journal}{\emph{Comput. Oper. Res.}}  \bibinfo{volume}{92}
  (\bibinfo{year}{2018}), \bibinfo{pages}{37--46}.
\newblock


\bibitem[\protect\citeauthoryear{Seidman}{Seidman}{1983}]%
        {83kcoredef}
\bibfield{author}{\bibinfo{person}{Stephen~B. Seidman}.}
  \bibinfo{year}{1983}\natexlab{}.
\newblock \showarticletitle{Network structure and minimum degree}.
\newblock \bibinfo{journal}{\emph{Social Networks}} \bibinfo{volume}{5},
  \bibinfo{number}{3} (\bibinfo{year}{1983}), \bibinfo{pages}{269--287}.
\newblock


\bibitem[\protect\citeauthoryear{Seidman and Foster}{Seidman and
  Foster}{1978}]%
        {seidman1978graph}
\bibfield{author}{\bibinfo{person}{Stephen~B. Seidman} {and}
  \bibinfo{person}{Brian~L. Foster}.} \bibinfo{year}{1978}\natexlab{}.
\newblock \showarticletitle{A graph-theoretic generalization of the clique
  concept}.
\newblock \bibinfo{journal}{\emph{Journal of Mathematical sociology}}
  \bibinfo{volume}{6}, \bibinfo{number}{1} (\bibinfo{year}{1978}),
  \bibinfo{pages}{139--154}.
\newblock


\bibitem[\protect\citeauthoryear{Tomita, Tanaka, and Takahashi}{Tomita
  et~al\mbox{.}}{2006}]%
        {TomitaTT06worstcaseclique}
\bibfield{author}{\bibinfo{person}{Etsuji Tomita}, \bibinfo{person}{Akira
  Tanaka}, {and} \bibinfo{person}{Haruhisa Takahashi}.}
  \bibinfo{year}{2006}\natexlab{}.
\newblock \showarticletitle{The worst-case time complexity for generating all
  maximal cliques and computational experiments}.
\newblock \bibinfo{journal}{\emph{Theor. Comput. Sci.}} \bibinfo{volume}{363},
  \bibinfo{number}{1} (\bibinfo{year}{2006}), \bibinfo{pages}{28--42}.
\newblock


\bibitem[\protect\citeauthoryear{Wang, Chen, Hou, Suo, Li, Pan, and Ives}{Wang
  et~al\mbox{.}}{2017}]%
        {WangCHSLPI17}
\bibfield{author}{\bibinfo{person}{Zhuo Wang}, \bibinfo{person}{Qun Chen},
  \bibinfo{person}{Boyi Hou}, \bibinfo{person}{Bo Suo},
  \bibinfo{person}{Zhanhuai Li}, \bibinfo{person}{Wei Pan}, {and}
  \bibinfo{person}{Zachary~G. Ives}.} \bibinfo{year}{2017}\natexlab{}.
\newblock \showarticletitle{Parallelizing maximal clique and k-plex enumeration
  over graph data}.
\newblock \bibinfo{journal}{\emph{J. Parallel Distributed Comput.}}
  \bibinfo{volume}{106} (\bibinfo{year}{2017}), \bibinfo{pages}{79--91}.
\newblock


\bibitem[\protect\citeauthoryear{Wei, Chen, and Tsai}{Wei
  et~al\mbox{.}}{2021}]%
        {tpdsWeiCT21}
\bibfield{author}{\bibinfo{person}{Yi{-}Wen Wei}, \bibinfo{person}{Wei{-}Mei
  Chen}, {and} \bibinfo{person}{Hsin{-}Hung Tsai}.}
  \bibinfo{year}{2021}\natexlab{}.
\newblock \showarticletitle{Accelerating the Bron-Kerbosch Algorithm for
  Maximal Clique Enumeration Using GPUs}.
\newblock \bibinfo{journal}{\emph{{IEEE} Trans. Parallel Distributed Syst.}}
  \bibinfo{volume}{32}, \bibinfo{number}{9} (\bibinfo{year}{2021}),
  \bibinfo{pages}{2352--2366}.
\newblock


\bibitem[\protect\citeauthoryear{Wu and Pei}{Wu and Pei}{2007}]%
        {07pakddWuP}
\bibfield{author}{\bibinfo{person}{Bin Wu} {and} \bibinfo{person}{Xin Pei}.}
  \bibinfo{year}{2007}\natexlab{}.
\newblock \showarticletitle{A Parallel Algorithm for Enumerating All the
  Maximal \emph{k} -Plexes}. In \bibinfo{booktitle}{\emph{PAKDD}},
  Vol.~\bibinfo{volume}{4819}. \bibinfo{pages}{476--483}.
\newblock


\bibitem[\protect\citeauthoryear{Zhou, Hu, Xiao, and Fu}{Zhou
  et~al\mbox{.}}{2021}]%
        {ZhouHXF21}
\bibfield{author}{\bibinfo{person}{Yi Zhou}, \bibinfo{person}{Shan Hu},
  \bibinfo{person}{Mingyu Xiao}, {and} \bibinfo{person}{Zhang{-}Hua Fu}.}
  \bibinfo{year}{2021}\natexlab{}.
\newblock \showarticletitle{Improving Maximum k-plex Solver via Second-Order
  Reduction and Graph Color Bounding}. In \bibinfo{booktitle}{\emph{{AAAI}}}.
\newblock


\bibitem[\protect\citeauthoryear{Zhou, Xu, Guo, Xiao, and Jin}{Zhou
  et~al\mbox{.}}{2020}]%
        {ZhouXGXJ20}
\bibfield{author}{\bibinfo{person}{Yi Zhou}, \bibinfo{person}{Jingwei Xu},
  \bibinfo{person}{Zhenyu Guo}, \bibinfo{person}{Mingyu Xiao}, {and}
  \bibinfo{person}{Yan Jin}.} \bibinfo{year}{2020}\natexlab{}.
\newblock \showarticletitle{Enumerating Maximal \emph{k}-Plexes with Worst-Case
  Time Guarantee}. In \bibinfo{booktitle}{\emph{{AAAI}}}.
  \bibinfo{pages}{2442--2449}.
\newblock


\end{thebibliography}
	
\end{document}